\documentclass{amsart}
\usepackage{amsfonts,amssymb}
\usepackage{amsmath,amscd}
\usepackage{graphicx}
\usepackage[enableskew]{youngtab}

\newtheorem{theorem}{Theorem}[section]

\newtheorem{claim}[theorem]{Claim}

\newtheorem{corollary}[theorem]{Corollary}

\newtheorem{definition}[theorem]{Definition}
\newtheorem{example}[theorem]{Example}

\newtheorem{lemma}[theorem]{Lemma}

\newtheorem{proposition}[theorem]{Proposition}
\newtheorem{remark}[theorem]{Remark}

\numberwithin{equation}{section}
\numberwithin{figure}{section}
\input{tcilatex}

\begin{document}
\title[Quantum cohomology via vicious and osculating walkers]{Quantum
cohomology via vicious and osculating walkers}
\author{Christian Korff }
\address{School of Mathematics and Statistics, University of Glasgow,\\
15 University Gardens, Glasgow G12 8QW, Scotland, UK }
\email{christian.korff@glasgow.ac.uk}

\begin{abstract}
We relate the counting of rational curves intersecting Schubert varieties of
the Grassmannian to the counting of certain non-intersecting lattice paths
on the cylinder, so-called vicious and osculating walkers. These lattice
paths form exactly solvable statistical mechanics models and are obtained
from solutions to the Yang-Baxter equation. The eigenvectors of the transfer
matrices of these models yield the idempotents of the Verlinde algebra of
the gauged $\mathfrak{\hat{u}}(n)_{k}$-WZNW model. The latter is known to be
closely related to the small quantum cohomology ring of the Grassmannian. We
establish further that the partition functions of the vicious and osculating
walker model are given in terms of Postnikov's toric Schur functions and can
be interpreted as generating functions for Gromov-Witten invariants. We
reveal an underlying quantum group structure in terms of Yang-Baxter
algebras and use it to give a generating formula for toric Schur functions
in terms of divided difference operators which appear in known
representations of the nil-Hecke algebra.
\end{abstract}

\maketitle

\smallskip \noindent \textbf{2010 Mathematics Subject Classification.}
14N35,05E05,05A15,05A19,82B23.\smallskip\newline
\noindent \textbf{Keywords.} Gromov-Witten invariants, quantum cohomology,
enumerative combinatorics, exactly solvable models, Bethe ansatz.






\section{Introduction}

Let $\func{Gr}\nolimits_{n,n+k}$ be the Grassmannian of $n$-planes in $%
\mathbb{C}^{n+k}$ and consider its small quantum cohomology ring $qH^{\ast }(%
\func{Gr}\nolimits_{n,n+k})$. The latter has the following presentation \cite%
{SiebertTian} 
\begin{equation}
qH^{\ast }(\func{Gr}\nolimits_{n,n+k})\cong \mathbb{Z}[q][e_{1},\ldots
,e_{n},h_{1},\ldots ,h_{k}]/\mathcal{I},  \label{qH*}
\end{equation}%
where the two-sided ideal $\mathcal{I}$ is generated by the coefficients of
the following polynomial in the auxiliary variable $x$,%
\begin{equation}
\left( \tsum_{i=0}^{n}e_{i}x^{i}\right) \left(
\tsum_{j=0}^{k}h_{j}x^{j}\right) =1+(-1)^{n}qx^{N}\;.  \label{Ipoly0}
\end{equation}%
Denoting by $(n,k)$ the set of all partitions whose Young diagram fits into
the $n\times k$ bounding box, a vector space basis of $qH^{\ast }(\func{Gr}%
\nolimits_{n,n+k})$ is given by the finite set $\{s_{\lambda }\}_{\lambda
\in (n,k)}$ of Schubert classes $s_{\lambda }=\det (h_{\lambda
_{i}-i+j})_{1\leq i,j\leq n}=\det (e_{\lambda _{i}^{\prime }-i+j})_{1\leq
i,j\leq k}$, where $\lambda ^{\prime }$ is the conjugate partition of $%
\lambda $. Note that this definition includes the special cases $%
s_{(r)}=h_{r}$ and $s_{(1^{r})}=e_{r}$.

\subsection{Quantum Kostka numbers}

Quantum Kostka numbers were originally introduced in \cite{BCF} as the
coefficients in the following product expansion in (\ref{qH*}),%
\begin{equation}
s_{\mu }\ast s_{\lambda _{1}}\ast \cdots \ast s_{\lambda _{r}}=\sum_{d\geq
0,\nu \in (n,k)}q^{d}s_{\nu }K_{\nu /d/\mu ,\lambda },  \label{qKostkadef}
\end{equation}%
where $\mu \in (n,k)$ and $\lambda =(\lambda _{1},\ldots ,\lambda _{r})$ are
some non-negative integers $\leq k$. As explained in \cite{BCF} and \cite%
{Postnikov} the quantum Kostka numbers count certain cylindric skew
tableaux, a combinatorial notion first introduced by Gessel and
Krattenthaler \cite{GesselKrattenthaler}. Analogously, one can define
conjugate quantum Kostka numbers by considering the product expansion%
\begin{equation}
s_{\mu }\ast s_{(1^{\lambda _{1}})}\ast \cdots \ast s_{(1^{\lambda
_{r}})}=\sum_{d\geq 0,\nu \in (n,k)}q^{d}s_{\nu }K_{\nu ^{\prime }/d/\mu
^{\prime },\lambda },  \label{qKostkadef2}
\end{equation}%
where $0\leq \lambda _{i}\leq n$. Due to level-rank duality, $qH^{\ast }(%
\func{Gr}\nolimits_{n,n+k})\cong qH^{\ast }(\func{Gr}\nolimits_{k,n+k})$,
one has $K_{\nu ^{\prime }/d/\mu ^{\prime },\lambda }=K_{\nu /d/\mu ,\lambda
}$. Nevertheless, both quantum product expansions give rise to different
combinatorics, so we will consider them separately.

Exploiting the quantum Giambelli formula \cite{Bertram}, $s _{\lambda }=\det
( s _{\lambda _{i}-i+j})_{1\leq i,j\leq n}$ with $\lambda \in (n,k)$, one
can then compute products between arbitrary Schubert classes%
\begin{equation}
s _{\mu }\ast s _{\lambda }=\sum_{d\geq 0,\nu \in (n,k)}q^{d}C_{\lambda \mu
}^{\nu ,d} s _{\nu },  \label{GWinv}
\end{equation}%
where the expansion coefficients are the 3-point, genus 0 Gromov-Witten
invariants. The latter count rational curves of degree $d=(|\lambda |+|\mu
|-|\nu |)/N$ intersecting three Schubert varieties in general position; see
e.g. \cite{Bertram}, \cite{AW}, \cite{BCF}, \cite{Buch}, \cite{Buchetal}, 
\cite{FW} for details.

\subsection{Toric Schur polynomials and Frobenius structures}

Quantum cohomology had its origin in mathematical physics and appeared first
in works of Gepner \cite{Gepner}, Intriligator \cite{Intriligator}, Vafa 
\cite{Vafa} and Witten \cite{Witten} in connection with the fusion ring $%
\mathcal{F}_{n,k}$\ of the gauged $\widehat{\mathfrak{u}}(n)_{k}$
Wess-Zumino-Novikov-Witten (WZNW) model. It has apparently been proved in
the no longer publicly available work \cite{Agnihotri} that $\mathcal{F}%
_{n,k}\cong qH^{\ast }(\func{Gr}\nolimits_{n,n+k})/\langle q-1\rangle $. $%
\mathcal{F}_{n,k}^\mathbb{C} =\mathcal{F}_{n,k}\otimes_{\mathbb{Z}} \mathbb{C%
}$ is referred to as \emph{Verlinde algebra} in the literature.

Denote by $\lambda ^{\vee }$ the partition obtained by taking the complement
of the Young diagram of $\lambda $ in the $n\times k$ bounding box and
recall the intersection pairing of Schubert classes, $\eta( s _{\lambda }, s
_{\mu }):= \int_{\func{Gr}\nolimits_{n,n+k}} s_\lambda\cdot s_\mu=\delta
_{\lambda ^{\vee }\mu }$.

\begin{proposition}
\label{Frob2toric} $\mathcal{F}_{n,k}^{\mathbb{C}}$ with bilinear form $\eta 
$ is a commutative Frobenius algebra and its coproduct is given by%
\begin{equation}
\Delta _{n,k}s_{\nu }=\sum_{d\geq 0,\mu \in (n,k)}q^{d}s_{\nu /d/\mu
}\otimes s_{\mu },\quad s_{\nu /d/\mu }:=\sum_{\lambda \in (n,k)}C_{\lambda
\mu }^{\nu ,d}s_{\lambda }\;.  \label{cop}
\end{equation}
\end{proposition}

Setting $e_{r}=\sum_{1\leq i_{1}<\cdots <i_{r}\leq n}x_{i_{1}}\cdots
x_{i_{r}}$ and $h_{r}=\sum_{1\leq i_{1}\leq \cdots \leq i_{r}\leq
n}x_{i_{1}}\cdots x_{i_{r}}$ where $x=(x_{1},\ldots ,x_{n})$ are some
commuting indeterminates the Schubert classes can be identified with \emph{%
Schur functions}. In light of (\ref{cop}) it is then natural to consider
so-called \emph{toric Schur functions} \cite{Postnikov}%
\begin{equation}
s_{\nu /d/\mu }(x)=\sum_{\lambda \in (n,k)}C_{\lambda \mu }^{\nu
,d}s_{\lambda }(x)=\sum_{\lambda \in (n,k)}K_{\nu /d/\mu ,\lambda
}m_{\lambda }(x)\;.  \label{toricSchur0}
\end{equation}%
The last identity can be seen as a combinatorial definition of $s_{\nu
/d/\mu }$ in terms of monomial symmetric functions; it is the weighted sum
over all toric skew tableaux, which is the subset of the cylindric skew
tableaux having at most $k$ boxes in each row. This generalises the notion
of an ordinary skew Schur function, $s_{\nu /\mu }=\sum_{\lambda }K_{\nu
/\mu ,\lambda }s_{\lambda }$, where the sum is over all skew (semi-standard)
tableaux and $K_{\nu /\mu ,\lambda }$ is the ordinary Kostka number; see
e.g. \cite{Macdonald}. Since for vanishing degree $d$ the Gromov-Witten
invariants equal the Littlewood-Richardson coefficients, one has $s_{\nu
/0/\mu }=s_{\nu /\mu }$ when $d=0$. In the case of infinitely many variables
one obtains cylindric Schur functions which have been investigated in \cite%
{GesselKrattenthaler}, \cite{McNamara} and in \cite{Lam} as special case of
affine Stanley symmetric functions; see also \cite{Borodin} for a
formulation of a random process on cylindric partitions. For a
generalisation of cylindric or toric Schur functions to special cases of
Macdonald functions in the context of the $\widehat{\mathfrak{su}}(n)_{k}$
fusion ring, see \cite{Korffmac}.

\subsection{Exactly solvable lattice models: vicious and osculating walkers}

\begin{figure}[tbp]
\begin{equation*}
\includegraphics[scale=0.26]{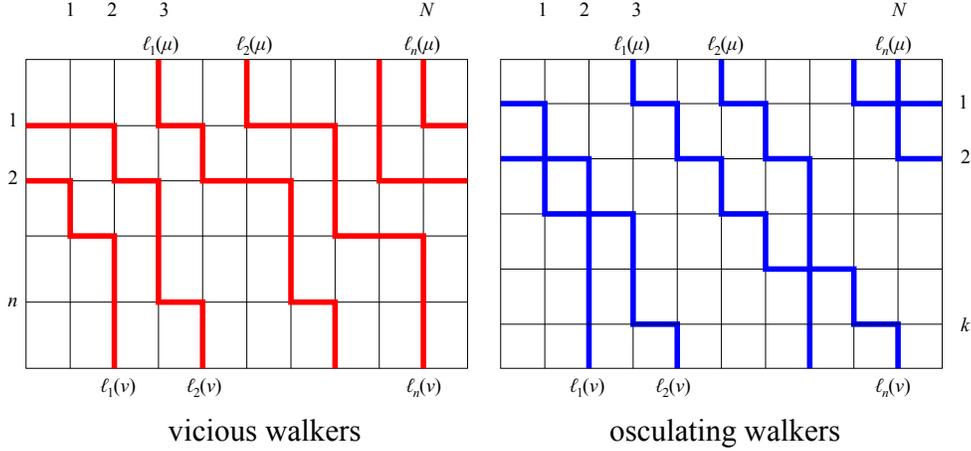}
\end{equation*}%
\caption{Examples of vicious and osculating walks on a square lattice with $%
N=n+k=4+5=9$. The paths of two osculating walkers can touch but they do not
intersect.}
\label{fig:walks}
\end{figure}

In this article we identify the toric Schur polynomial (\ref{toricSchur0})
with the partition function of exactly solvable lattice models in
statistical mechanics, the lock-step vicious and osculating walker models
which have appeared in connection with problems such as percolation in
physical systems \cite{Fisher} and the counting of alternating sign matrices 
\cite{Brak}. Both models can be formulated in terms of special
non-intersecting paths on a square lattice which in the present context we
choose to have dimensions $n\times (n+k)$ and $k\times (n+k)$; see Figure %
\ref{fig:walks} for examples. Fixing start and end positions of the walkers
in terms of partitions $\mu ,\nu \in (n,k)$ and identifying the left with
the right lattice edge, we show that there is a bijection between toric
tableaux of shape $\nu /d/\mu $ and non-intersecting paths of the mentioned
models on the cylinder. The degree $d\in\mathbb{Z}_{\geq 0}$ is the number
of walkers crossing the boundary, where the horizontal strip is glued
together to obtain the cylinder. The weight $\lambda$ of a toric tableau
fixes the number of horizontal path edges in each lattice row.

Denote by$\ H=(H_{\nu ,\mu })$ and $E=(E_{\nu ,\mu })$ the transfer matrices
whose elements are the partition functions of a single lattice row for the
vicious and osculating walker models on the cylinder with fixed start and
end configurations $\mu ,\nu \in (n,k)$. Then the matrix elements of the
powers $H^{n}$ and $E^{k}$ give the partition functions for the lattices
with $n$ and $k$ rows mentioned above.

Given a square $s=\langle i,j \rangle$ in the Young diagram of a partition $%
\lambda $ denote by $c(s)=j-i$ its content and by $h(s)=\lambda_i+\lambda^{%
\prime}_j-i-j+1$ its hook length.

\begin{theorem}
\label{sumrule} For fixed start and end points $\mu ,\nu \in (n,k)$, the
number of vicious and osculating walker configurations on the cylinder is
given by%
\begin{equation}
(H^{n})_{\nu ,\mu }=\sum_{d,~\lambda }K_{\nu /d/\mu ,\lambda }=(E^{n})_{\nu
^{\prime },\mu ^{\prime }}  \label{sumrule00}
\end{equation}%
or alternatively%
\begin{equation}
(H^{n})_{\nu ,\mu }=\sum_{d,~\lambda }C_{\lambda \mu }^{\nu ,d}\prod_{s\in
\lambda }\frac{n+c(s)}{h(s)}=(E^{n})_{\nu ^{\prime },\mu ^{\prime }}\;,
\label{sumrule0}
\end{equation}%
where the sums run over all integers $d\geq 0$ and $\lambda \in (n,k)$ such
that $|\lambda |+|\mu |-|\nu |=d(n+k)$.
\end{theorem}

Note that (\ref{sumrule0}) provides a linear system of equattions for
Gromov-Witten invariants.

It is not obvious from their definition, but we will show that the
row-to-row transfer matrices $E$ and $H$ commute and possess a common
eigenbasis $\{\mathfrak{\hat{e}}_{\lambda }\}_{\lambda \in (n,k)}$, the
so-called Bethe vectors. They yield the idempotents of the Verlinde algebra $%
\mathcal{F}_{n,k}^{\mathbb{C}}$ of the gauged $\mathfrak{\hat{u}}(n)_{k}$%
-WZNW model.

\begin{theorem}[\textbf{idempotents of the Verlinde algebra}]
Let $V_{n,k}$ be the complex linear span of the Bethe vectors $\{\mathfrak{%
\hat e}_\lambda\}_{\lambda\in(n,k)}$. The generalised matrix algebra $%
(V_{n,k},\star)$ obtained by setting $\mathfrak{\hat e}_\lambda \star%
\mathfrak{\hat e}_\mu=\delta_{\lambda\mu}\mathfrak{\hat e}_\lambda$ is
isomorphic to $\mathcal{F}_{n,k}^\mathbb{C}$.
\end{theorem}

The same basis has been employed in \cite[Thm 10.11]{KS} to provide an
alternative derivation of the presentation (\ref{qH*}). They are identical
with the Bethe vectors of the so-called XX-Heisenberg spin chain \cite[Rm
10.3]{KS} and, thus, one can identify the ring (\ref{qH*}) with the
conserved quantities or \textquotedblleft quantum integrals of
motion\textquotedblright of this quantum spin-chain. There are close
parallels with recent developments regarding the link between topological
field theories and quantum integrable models \cite{NekrasovShatashvili} as
well as the quantum cohomology of the \emph{cotangent bundles} of Grassmann
varieties \cite{Rimanyietal}, \cite{Gorbetal}.\smallskip

The new result in this article is the connection of the quantum cohomology
of the Grassmann varieties themselves with the mentioned statistical lattice
models which allows one to relate the counting of lattice paths on the
cylinder to Gromov-Witten invariants and to reveal a deeper underlying
algebraic structure which we now explain.

\subsection{Quantum group structures}

The combinatorial connection with the mentioned exactly solvable lattice
models is underpinned by an algebraic description known as the \emph{quantum
inverse scattering method} which is the name of a general procedure based on
the works of the Faddeev School; see e.g. \cite{KIB} and references therein.
Here we show that this method can be applied to the quantum cohomology ring.
Our starting point is a solution to the Yang-Baxter equation%
\begin{equation}
R_{12}(x/y)M_{1}(x)M_{2}(y)=M_{2}(y)M_{1}(x)R_{12}(x/y)  \label{YBE0}
\end{equation}%
where $R(x/y)\in \mathbb{C}(x/y)\otimes \limfunc{End}(V\otimes V)$ with $V$
a two-dimensional complex vector space and 
\begin{equation}
M(x)=\left( 
\begin{array}{cc}
A(x) & B(x) \\ 
C(x) & D(x)%
\end{array}%
\right) \in \mathbb{C}[x]\otimes \limfunc{End}V\otimes \limfunc{End}%
V^{\otimes N}  \label{mom0}
\end{equation}%
is the so-called monodromy matrix. The matrix entries of $M(x)$ can be
interpreted as vertex-type operators whose commutation relations are encoded
in the matrix $R$. The latter generate the so-called Yang-Baxter algebra
which has the structure of a graded bi-algebra. In particular one can
interpret (\ref{YBE0}), (\ref{mom0}) as the defining relations of a
\textquotedblleft quantum group\textquotedblright ; this is similar to the
construction of a Yangian symmetry on the quantum cohomologies of cotangent
bundles of Nakajima varieties in a recent preprint by Maulik and Okounkov 
\cite{MaulikOkounkov}.

Set $N=n+k$ and for simplicity denote $qH^{\ast }(\limfunc{Gr}_{n,k})\otimes
_{\mathbb{Z}}\mathbb{C}$ by $qH_{n,k}^{\ast }$. Then there exists a vector
space isomorphism 
\begin{equation}
V^{\otimes N}\otimes \mathbb{C}[q,q^{-1}]\overset{\sim }{\rightarrow }%
\tbigoplus_{n=0}^{N}qH_{n,N-n}^{\ast }  \label{veciso}
\end{equation}%
which induces maps $A_{r},D_{r}:qH_{n,k}^{\ast }\rightarrow qH_{n,k}^{\ast }$
and $B_{r},C_{r}:qH_{n,k}^{\ast }\rightarrow qH_{n\pm 1,k\mp 1}^{\ast }$
where $A_{r},B_{r},C_{r},D_{r}$ are the coefficients of $x^{r}$ in the
vertex-type operators (\ref{mom0}) and $qH_{0,N}^{\ast },qH_{N,0}^{\ast }%
\mathbb{\cong C}$. Exploiting level-rank duality, $\Theta :qH_{n,k}^{\ast }%
\overset{\sim }{\rightarrow }qH_{k,n}^{\ast },$ one finds a second, dual
solution $M^{\prime }=\Theta \circ M\circ \Theta $ of the Yang-Baxter
equation. The following result then links these two Yang-Baxter algebras
with the Frobenius algebra $qH_{n,k}^{\ast }$ via (\ref{Ipoly0}).

\begin{theorem}
Setting $H(x)=A(x)+qD(x)$ and $E(x)=A^{\prime }(x)+qD^{\prime }(x)$, the
restrictions $E(x)|_{qH_{n,k}^{\ast }}$ and $H(x)|_{qH_{n,k}^{\ast }}$ under
(\ref{veciso}) are of polynomial degree $n$ and $k$, respectively and on
each preimage of $qH_{n,k}^{\ast }$ one has the functional identity%
\begin{equation}
E(-x)H(x)=1+(-1)^{n}qx^{N}\;.  \label{QQ0}
\end{equation}%
In particular, the map given by $E_{i}=A_{i}^{\prime }+qD_{i}^{\prime
}\mapsto e_{i}$ and $H_{j}=A_{j}+qD_{j}\mapsto h_{j}$ with $i=1,\ldots ,n$
and $j=1,\ldots ,k$ yields an algebra isomorphism $\mathcal{A}_{n,k}\overset{%
\sim }{\rightarrow }qH_{n,k}^{\ast }$, where $\mathcal{A}_{n,k}\subset 
\limfunc{End}(qH_{n,k}^{\ast })$ is the commutative subalgebra generated by
the coefficients of the Yang-Baxter algebra elements $E(x)$ and $H(x)$.
\end{theorem}

The combinatorial results (\ref{sumrule00}) and (\ref{sumrule0}) then follow
from the fact that one recovers the row-to-row vicious and osculating walker
transfer matrices at $x=1$, i.e. $H=H(1)$ and $E=E(1)$. An alternative way
to express the relation between the Yang-Baxter algebras and $qH_{n,k}^{\ast
}$ is to use the coproduct of the Frobenius algebra.

\begin{proposition}
\label{Prop:cop}Setting as before $H(x)=A(x)+qD(x)$ one has for any $1\leq
n\leq N-1$ that%
\begin{equation}
\langle \lambda |H(y_{n})\cdots H(y_{1})B(x_{n})\cdots B(x_{1})|0\rangle
=x^{\delta _{n}}\Delta _{n,k}s_{\lambda }(y,x)  \label{cop2}
\end{equation}%
where $\langle \lambda |$ denotes the dual basis of the Schubert classes
under the isomorphism (\ref{veciso}), $|0\rangle $ the unique basis vectors
in $qH_{0,N}^{\ast }$, $\delta _{n}=(n,n-1,\ldots ,1)$ and $\Delta _{n,k}$
is the coproduct (\ref{cop}) of $qH_{n,k}^{\ast }$.
\end{proposition}

This last results implies in particular that one can use the commutation
relations of the Yang-Baxter algebra encoded in (\ref{YBE0}) to compute the
toric Schur functions (\ref{toricSchur0}) by moving the $H$-operators past
the $B$-operators. We will thus derive the following formula in terms of
divided difference operators (also called Demazure or
Bernstein-Gelfand-Gelfand operators).

Set $\partial _{i}=(1-x_{i+1}/x_{i})^{-1}(s_{i}-1)$ with $s_{i}$ being the
transposition which acts by switching $x_{i}$ with $x_{i+1}$ and define $%
\nabla _{i}=\partial _{n-1+i}\cdots \partial _{i+1}\partial _{i}$.

\begin{corollary}
\label{Cor:Demazure}Choose $y_{i}=x_{i+n}$ in (\ref{cop2}) and introduce the
\textquotedblleft generating function\textquotedblright\ 
\begin{equation}
F_{\lambda }(x;y)=s_{\lambda }(y)\prod_{i=1}^{n}(1+qx_{i}^{N})\;.
\label{toric_generating}
\end{equation}%
Then we have the following formula in terms of Demazure operators,%
\begin{equation}
\sum_{d\geq 0}q^{d}s_{\lambda /d/\mu }(y)=\left\langle s_{\mu
}(x)||x^{-\delta _{n}}\nabla _{1}\nabla _{2}\cdots \nabla _{n}y^{\delta
_{n}}F_{\lambda }(x;y)\right\rangle  \label{toricDemazure}
\end{equation}%
where the notation $\left\langle s_{\mu }(x)||\cdots \right\rangle $ denotes
the coefficient of the Schur function $s_{\mu }(x)$.
\end{corollary}

We will demonstrate on a simple example how toric Schur functions can be
explicitly computed by invoking the last formula in Section \ref%
{Sec:Demazure}.

Another identity for toric Schur functions - which is a direct consequence
of the quantum group structure and does not seem to have appeared previously
in the literature - is the following sum rule.

\begin{corollary}
Let $N=2n$. Then we have the following identities for toric Schur functions%
\begin{equation}
e_{r}(x_{1}^{N},\ldots ,x_{n}^{N})=\sum_{\substack{ d+d^{\prime }=r  \\ \mu
\in (n,n)}}(-1)^{|\lambda |-|\mu |}s_{\lambda ^{\prime }/d^{\prime }/\mu
^{\prime }}(x_{1},\ldots ,x_{n})s_{\mu /d/\lambda }(x_{1},\ldots ,x_{n}),
\end{equation}%
where $\lambda $ is any partition in the $n\times n$ square.
\end{corollary}

This identity is a true \textquotedblleft quantum
relation\textquotedblright\ as it becomes trivial for $q=0$, i.e. there
exists no analogue of this relation for skew Schur functions where $%
d,d^{\prime }=0$. Similar identities hold also for $N\notin 2\mathbb{N}$ but
look more complicated. They will be stated in Section \ref{Sec:Demazure}.

\subsection{Outline of the article}

In Section 2 we introduce some preliminary combinatorial notions regarding
01-words and partitions.

In Section 3 we discuss in detail the vicious and osculating walker models.
While these have been introduced in the literature previously, our
conventions differ from the usual ones by rotating the lattice $45^\circ$
and choosing a special set of weights. We also analyse in depth the related
Yang-Baxter algebras and show that both models are related via level-rank
duality. We derive a matrix functional equation relating the transfer
matrices of both models.

Section 4 contains the algebraic Bethe ansatz construction of the
idempotents of the Verlinde algebra. As a byproduct of this construction we
obtain novel expressions for Schur functions as matrix elements of the above
mentioned Yang-Baxter algebras.

Section 5 states explicit bijections between vicious and osculating walker
configurations on the cylinder and toric tableaux, which can be interpreted
as a special subset of semi-standard tableau of skew shape. As a corollary
we obtain that Postnikov's toric Schur polynomials are the partition
functions of the specialised vicious and osculating walker models. The sum
rule (\ref{sumrule0}), relating Gromov-Witten invariants to the counting of
vicious and osculating walker configurations on the cylinder, is then an
immediate consequence.

We will conclude with stating the proofs of the various identities for toric
Schur functions arising from the quantum group structure.\medskip

\noindent \textbf{Acknowledgement}. The research presented in this article
has in part been financially supported by a University Research Fellowship
of the Royal Society, by the Engineering and Physical Sciences Research
Council grants EP/I037636/1 and GR/R93773/01 as well as the trimester
programme \emph{Integrability in Geometry and Mathematical Physics} (1
January -- 30 April 2012) at the Hausdorff Research Institute for
Mathematics, Bonn, Germany. Some of the results have been presented during
the trimester workshop \emph{Integrability in Topological Field Theory}
(16-20 April 2012) organised by Albrecht Klemm and Katrin Wendland. The
author wishes to thank the trimester organisers, Franz Pedit, Ulrich
Pinkall, Iskander Taimanov, Alexander Veselov and Katrin Wendland and the
staff of the Hausdorff Institute for their kind invitation and hospitality.
It is a pleasure to thank Vassily Gorbounov and Catharina Stroppel for
discussions.

\section{Preliminaries}

Throughout this article we consider non-negative integers $N,n,k\in \mathbb{Z%
}_{\geq 0}$ such that $N=n+k$ and set $I:=\{1,\ldots ,N\}$. Let $V=\mathbb{C}%
v_{0}\oplus \mathbb{C}v_{1}$ and denote by $V^{\ast }$ its dual. Consider
the tensor product $V^{\otimes N}$. We identify the standard basis $\mathcal{%
B}=\{v_{w_{1}}\otimes \cdots \otimes v_{w_{N}}:w_{i}=0,1\}\subset V^{\otimes
N}$ with the set of 01-words of length $N$, $W=\{w=w_{1}w_{2}\ldots
w_{N}:w_{i}=0,1\}$, via the map 
\begin{equation}
w\mapsto |w\rangle :=v_{w_{1}}\otimes \cdots \otimes v_{w_{N}}\;.
\label{word2b}
\end{equation}%
For convenience we are employing the Dirac notation and denote by $\langle
w| $ the dual basis with $\langle w|\tilde{w}\rangle =\prod_{i=1}^{N}\delta
_{w_{i},\tilde{w}_{i}}$. Furthermore, we shall denote by $W_{n}=\{w\in
W:|w|=\sum_{i=1}^{N}w_{i}=n\}$ the subset of all 01-words with $n$
one-letters, by $\mathcal{B}_{n}\subset V^{\otimes N}$ its image under the
above map (\ref{word2b}) and by $V_{n}\subset V^{\otimes N}$ the subspace
spanned by the corresponding basis elements in $\mathcal{B}_{n}$. As (\ref%
{word2b}) is a bijection we can also introduce the inverse map whose image
we denote by $w(b)$ with $b\in \mathcal{B}_{n}$.

There are alternative descriptions of the elements in $\mathcal{B}_{n}$
which will be useful for our discussion. Namely, consider the set of
partitions $\lambda $ whose Young diagram fits into a bounding box of height 
$n$ and width $k$; we shall denote it by $(n,k)$. Define a bijection $%
(n,k)\rightarrow W_{n}$ via 
\begin{equation}
\lambda \mapsto w(\lambda )=0\cdots 0\underset{\ell _{1}}{1}0\cdots 0%
\underset{\ell _{n}}{1}0\cdots 0,\qquad \ell _{i}(\lambda )=\lambda
_{n+1-i}+i\;,  \label{part2word}
\end{equation}%
where $\ell (\lambda )=(\ell _{1},\ldots ,\ell _{n})$ with $1\leq \ell
_{1}<\ldots <\ell _{n}\leq N$ denote the positions of 1-letters in $%
w(\lambda )$ from \emph{left to right}. We assume the latter to be periodic,
that is we set $\ell _{i+n}(\lambda )=\ell _{i}(\lambda )+N$. We shall
denote the image of the inverse of the map (\ref{part2word}) by $\lambda (w)$
and by $|\lambda \rangle $ the corresponding ket vector in $\mathcal{B}_{n}$%
. Note that the correspondence (\ref{part2word}) can be easily understood
graphically: the Young diagram of the partition $\lambda $ traces out a path
in the $n\times k$ rectangle which is encoded in $w$. Starting from the left
bottom corner in the $n\times k$ rectangle go one box right for each letter $%
0$ and one box up for each letter 1; see Figure \ref{fig:01worddict} for an
example. 
For later purposes we introduce the notation 
\begin{equation}
n_{\ell }(\lambda )=\sum_{i=1}^{\ell }w_{i}(\lambda )  \label{no}
\end{equation}%
for $\ell \in I$ and set $n_{\ell +N}(\lambda )=n_{\ell }(\lambda )+n$ for
any $\ell \in \mathbb{Z}$.

Exploiting the bijection (\ref{part2word}) there are two operations on
01-words which are induced by taking the complement of $\lambda \in (n,k)$
in the bounding box, $\lambda \mapsto \lambda ^{\vee }:=(k-\lambda
_{n},\ldots ,k-\lambda _{2},k-\lambda _{1})$, and by considering the
conjugate partition $\lambda ^{\prime }\in (k,n)$. The corresponding
01-words $w(\lambda ^{\vee }),w(\lambda ^{\prime })$ are obtained from $%
w(\lambda )$ via the maps 
\begin{equation}
w\mapsto w^{\vee }:=w_{N}\ldots w_{2}w_{1}\quad \text{and}\quad w\mapsto
w^{\prime }=(1-w_{N})\cdots (1-w_{2})(1-w_{1}),  \label{01bijections}
\end{equation}%
respectively. Note that these maps yield bijections $W_{n}\rightarrow W_{n}$
and $W_{n}\rightarrow W_{k}$. We will also make use of the combined map $%
w\mapsto w^{\#}:=(w^{\vee })^{\prime }=(w^{\prime })^{\vee }$ which is
simply the exchange of 0 and 1-letters.

\begin{figure}[tbp]
\begin{equation*}
\includegraphics[scale=0.26]{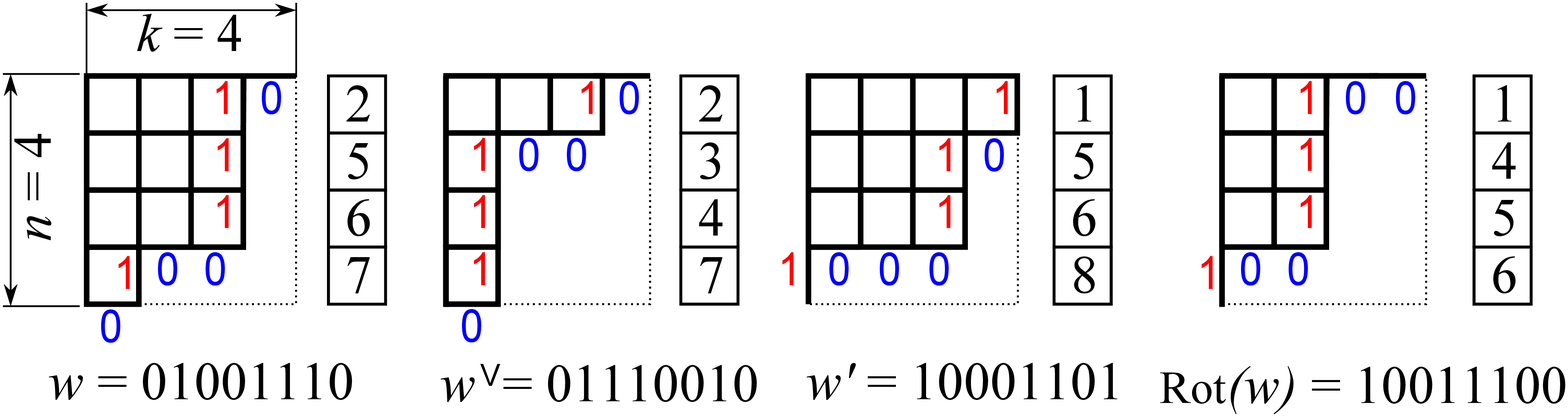}
\end{equation*}%
\caption{01-words, partitions and single column tableaux. Depicted are the
various transformations of 01-word defined in the text for $N=2n=2k=8$.}
\label{fig:01worddict}
\end{figure}

There is one additional map $\func{Rot}:W_{n}\rightarrow W_{n}$ which we
require for our discussion: set %
$w\mapsto \func{Rot}(w):=w_{2}w_{3}\ldots w_{N}w_{1}\,$which translates via (%
\ref{part2word}) to the action%
\begin{equation}
\func{Rot}(\lambda ):=\left\{ 
\begin{array}{cc}
(\lambda _{1}-1,\ldots ,\lambda _{n}-1), & \lambda _{n}>0 \\ 
(k,\lambda _{1},\ldots ,\lambda _{n-1}), & \text{else}%
\end{array}%
\right. \;.  \label{Rotdef}
\end{equation}%
Exploiting the last expression one then derives the following formula with $%
|\lambda|=\sum_i\lambda_i$, 
\begin{equation}
|\func{Rot}\nolimits^{\ell }(\lambda )|=|\lambda |-\ell n+n_{\ell }(\lambda
)N\;.  \label{Rotabs}
\end{equation}%
Note that obviously we have $\func{Rot}\nolimits^{N}(\lambda )=\lambda $.
For obvious reasons we will refer to $\func{Rot}$ as the \emph{rotation
operator}. The maps (\ref{01bijections}) and (\ref{Rotdef}) are significant
for our discussion as they constitute symmetries of Gromov-Witten invariants.

\section{Vicious and osculating walkers}

We recall the definition of the lockstep vicious walker model originally
introduced by Fisher \cite{Fisher} and show that this statistical mechanics
model with the correct choice of weights and boundary conditions is closely
related to the small quantum cohomology ring of the Grassmannian: its
partition function can be interpreted as generating function of 3-point,
genus zero Gromov-Witten invariants.

There is another statistical model introduced by Brak \cite{Brak}, called
osculating walkers\footnote{%
We note that Brak's model is a six-vertex model, i.e. has different
Boltzmann weights from the one discussed here and in particular has one more
allowed vertex configuration. However, the crucial vertex configuration with
two paths approaching each other arbitarily close is also present here and
we therefore adopt his nomenclature; see Figure \ref{fig:QCvertex2}.}, which
in our setting turns out to be dual or complementary to the vicious walker
model. Namely, we will show that the transfer matrices of the vicious and
osculating walker models are given in terms of analogues of complete and
elementary symmetric functions in certain noncommutative variables: the
generators of the nil affine Temperely-Lieb algebra.

\subsection{Vicious walkers: vertex and lattice configurations}

We start with the 5-vertex formulation of the vicious walker model. Fix two
integers $N>0$ and $0\leq n\leq N$ and consider the square lattice 
\begin{equation}
\mathbb{L}:=\{\langle i,j\rangle \in \mathbb{Z}^{2}|0\leq i\leq n+1,\;0\leq
j\leq N+1\}\,.  \label{lattice}
\end{equation}%
Denote by $\mathbb{E}=\{(p,p^{\prime })\in \mathbb{L}^{2}:p_{1}+1=p_{1}^{%
\prime },\,p_{2}=p_{2}^{\prime }\,\text{or}\,p_{1}=p_{1}^{\prime
},\,p_{2}+1=p_{2}^{\prime }\}$ the set of horizontal and vertical edges.

\begin{definition}
A \emph{lattice configuration}\ $\mathcal{C}:\mathbb{E}\rightarrow \{0,1\}$
is an assignment of values $0$ or $1$ to the lattice edges.
\end{definition}

The weight of a configuration $\mathcal{C}$ is defined as the product over
its vertex weights, 
\begin{equation}
\func{wt}(\mathcal{C})=\prod_{(i,j)\in \mathbb{L}}\func{wt}(\mathrm{v}%
_{i,j})\in \mathbb{Z}[x_{1},\ldots ,x_{n}]\;,  \label{wtC}
\end{equation}%
where $\mathrm{v}_{i,j}$ denotes the vertex obtained by intersecting the $i^{%
\text{th}}$ horizontal lattice line with the $j^{\text{th}}$ vertical one.
That is, a vertex configuration is a 4-tuple $\mathrm{v}_{i,j}=(a,b,c,d)$
where respectively $a,b,c,d=0,1$ are the values of the W, N, E, S edges at
the lattice point $\langle i,j\rangle $. There are 5 allowed vertex
configurations which are depicted in Figure \ref{fig:QCvertex} together with
their weights. All other vertex configurations are forbidden, i.e. they have
weight zero. Some of the nonzero weights are given in terms of a set of
commutative indeterminates $(x_{1},\ldots ,x_{n})$, one for each row. 
\begin{figure}[tbp]
\begin{equation*}
\includegraphics[scale=0.35]{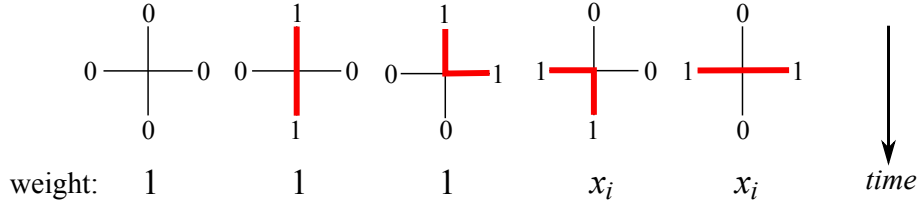}
\end{equation*}%
\caption{The 5 allowed vertex configurations and their corresponding vicious
walk sections.}
\label{fig:QCvertex}
\end{figure}

Connecting the 1-letters in each vertex configuration as shown in Figure \ref%
{fig:QCvertex}, it is easy to see that each lattice configuration
corresponds to a configuration of $n$ non-intersecting paths, where a path $%
\gamma =(p_{1},\ldots ,p_{l})$ is a sequence of points $p_{r}=(i_{r},j_{r})%
\in \mathbb{L}$ such that either $p_{r+1}=(i_{r}+1,j_{r})$ or $%
(i_{r},j_{r}+1)$, i.e. a connected sequence of horizontal and vertical edges
as depicted in Figure \ref{fig:walk_tableau_bijection}.

\subsection{Transfer matrix and nil Temperley-Lieb polynomials}

Define $\sigma ^{-}=\left( 
\begin{smallmatrix}
0 & 1 \\ 
0 & 0%
\end{smallmatrix}%
\right) ,\;\sigma ^{+}=\left( 
\begin{smallmatrix}
0 & 0 \\ 
1 & 0%
\end{smallmatrix}%
\right) ,\;\sigma ^{z}=\left( 
\begin{smallmatrix}
1 & 0 \\ 
0 & -1%
\end{smallmatrix}%
\right) $ to be the Pauli matrices acting on $V\cong \mathbb{C}^{2}$ via $%
\sigma ^{-}v_{1}=v_{0}$, $\sigma ^{+}v_{0}=v_{1}$ and $\sigma ^{z}v_{\alpha
}=(-1)^{\alpha }v_{\alpha },~\alpha =0,1$. We now interpret the possible
vertex configuration in the $i^{\text{th}}$ row and $j^{\text{th}}$ column
as a map $L(x_{i}):V_{i}(x_{i})\otimes V_{j}\rightarrow V_{i}(x_{i})\otimes
V_{j}$ with $V_{i}(x_{i})=V_{i}\otimes \mathbb{C}(x_{i})$ and $V_{i}\cong
V_{j}\cong V$ for all $\langle i,j\rangle \in \mathbb{L}$. We will therefore
drop the row and column labels and, in addition, often suppress the
dependence on the indeterminate $x_{i}$ in the notation. Thus, the values of
the horizontal edges label the basis vectors in $V_{i}$ while the values of
the vertical edges label the basis vectors in $V_{j}$. The mapping is from
the NW to the SE direction through the vertex. That is, label with $%
a,b,c,d=0,1$ the values of the edges in Figure \ref{fig:QCvertex} in
clockwise direction starting from the W edge. Interpret the corresponding
weight $L_{cd}^{ab}=\func{wt}(\mathrm{v}_{i,j})$ as the matrix element of
the map $L$, where we set $L_{cd}^{ab}=0$ whenever the vertex configuration
is not allowed. We then obtain 
\begin{equation}
L(x_{i})v_{a}\otimes v_{b}=\sum_{c,d=0,1}L_{cd}^{ab}(x_{i})v_{c}\otimes
v_{d}=x_{i}^{a}[v_{0}\otimes (\sigma ^{+})^{a}v_{b}+v_{1}\otimes \sigma
^{-}(\sigma ^{+})^{a}v_{b}]\;,  \label{5vLmatrix}
\end{equation}%
which can be rewritten in the basis $\{v_{0}\otimes v_{0},v_{0}\otimes
v_{1},v_{1}\otimes v_{0},v_{1}\otimes v_{1}\}$ as%
\begin{equation}
L(x_{i})=\left( 
\begin{array}{cccc}
1 & 0 & 0 & 0 \\ 
0 & 1 & x_{i} & 0 \\ 
0 & 1 & x_{i} & 0 \\ 
0 & 0 & 0 & 0%
\end{array}%
\right) \;.  \label{L4x4}
\end{equation}

\begin{proposition}
The 5-vertex L-matrix satisfies the Yang-Baxter equation,%
\begin{equation}
R_{12}(x,y)L_{13}(x)L_{23}(y)=L_{23}(y)L_{13}(x)R_{12}(x,y),  \label{YBE1}
\end{equation}%
where the matrix $R$ is given by%
\begin{equation}
R(x,y)=\left( 
\begin{array}{cccc}
1 & 0 & 0 & 0 \\ 
0 & 0 & 1 & 0 \\ 
0 & y/x & 1-y/x & 0 \\ 
0 & 0 & 0 & y/x%
\end{array}%
\right)\;.  \label{R}
\end{equation}
\end{proposition}

\begin{proof}
A straightforward computation.
\end{proof}

Note that the matrix $R(x,y)$ is non-singular for generic $x,y$, since $\det
R(x,y)=-y^{2}/x^{2}$. The solution $L$ to the Yang-Baxter equation can be
used to define an algebra in $\func{End}(V)^{\otimes N}\cong \func{End}%
(V^{\otimes N})$, called the \emph{Yang-Baxter algebra} which plays a
central role in the quantum inverse scattering method; see e.g. \cite{KIB}
for a textbook and references therein. In fact, the Yang-Baxter algebra
comes naturally equipped with a coproduct. First rewrite the $L$-matrix in
the form%
\begin{equation}
L(x)=\left( 
\begin{array}{cc}
1 & x\sigma ^{+} \\ 
\sigma ^{-} & x\sigma ^{-}\sigma ^{+}%
\end{array}%
\right)  \label{L}
\end{equation}%
where the matrix elements are polynomials in the indeterminate $x$ with
coefficients in $\func{End}V$. The following is a known result how to
introduce a bi-algebra structure on solutions of the Yang-Baxter equation;
we therefore omit the proof.

\begin{proposition}
Define implicitly a coproduct $\Delta :\func{End}(V)\rightarrow \func{End}%
(V)\otimes \func{End}(V)$ by setting $\Delta L(x):=L_{13}(x)L_{12}(x)$ and a
co-unit $\varepsilon :\func{End}(V)\rightarrow \mathbb{C}$ by $\varepsilon
(L(x))=\left( 
\begin{smallmatrix}
1 & 0 \\ 
0 & 1%
\end{smallmatrix}%
\right) $. Here the maps $\varepsilon ,\Delta $ act on the coefficients when
expanding with respect to the spectral variable $x$. The set of solutions of
(\ref{YBE1}) equipped with $\Delta ,\varepsilon $ forms a bialgebra, so in
particular $\Delta L$ is again a solution of (\ref{YBE1}).
\end{proposition}

Note that we do not have a Hopf algebra structure as the $L$-operator is not
invertible.

Repeatedly applying $\Delta $ and the ismorphism $\func{End}(V)\otimes \func{%
End}(W)\cong \func{End}(V\otimes W)$, one is led to consider the so-called 
\emph{monodromy matrix} in $\func{End}(V(x))\otimes \func{End}(V^{\otimes
N}) $%
\begin{equation}
T(x)=L_{0N}(x)\cdots L_{02}(x)L_{01}(x)=\left( 
\begin{array}{cc}
A(x) & B(x) \\ 
C(x) & D(x)%
\end{array}%
\right) \;.  \label{T}
\end{equation}%
From the Yang-Baxter equation one then deduces - among others - the
following commutation relations for the entries of the monodromy matrix,%
\begin{eqnarray}
A(x)A(y) &=&A(y)A(x),\;D(x)D(y)=D(y)D(x),\;  \label{YBA1a} \\
B(x)B(y) &=&\frac{y}{x}~B(y)B(x),\;C(x)C(y)=\frac{x}{y}~C(y)C(x)  \notag
\end{eqnarray}%
%
%
%
%
%
%
%
%
%
%
%
%
%
%
%
%
%
%
%
%
%
%
%
%
%
%
%
%
%
%
and%
\begin{gather}
xB(x)A(y)=xB(y)A(x)+(x-y)A(y)B(x),\qquad  \label{YBA1b} \\
yB(x)D(y)+(x-y)D(x)B(y)=yB(y)D(x)\;.  \notag
\end{gather}%
%
%
%
%
%
%
%
%
%
%
%
%
%
%
%
%
%
%
%
%
%
%
%
%
%
%
%
%
%
%
We now describe these commutation relations in terms of divided difference
operators. This will allow us in a subsequent section to derive the
generating function (\ref{toric_generating}) for toric Schur functions and
formula (\ref{toricDemazure}) mentioned in the introduction. We will also
use these relations below to construct eigenvectors of the vicious walker
transfer matrix.

Consider the polynomial ring $\mathcal{R}_{m}=\mathbb{Z}[x_{1},\ldots
,x_{m}] $ where the $x_{i}$'s are some commuting indeterminates. There is a
natural action of the symmetric group $S_{m}$ on $\mathcal{R}_{m}$ by
permuting the $x_{i}$'s where we denote by $\{s_{i}\}_{i=1}^{m-1}$ the
elementary transpositions which exchange $x_{i}$ and $x_{i+1}$. Intoduce the
difference operators $\partial _{i}=(1-x_{i+1}/x_{i})^{-1}(s_{i}-1)$ for $%
i=1,\ldots ,m$. Despite first appearance the latter map polynomials into
polynomials: because of linearity it suffices to consider the following
action on a monomial%
\begin{equation}
\partial _{i}x_{i}^{a}x_{i+1}^{b}=\left\{ 
\begin{array}{cc}
\sum_{r=0}^{b-a-1}x_{i}^{b-r}x_{i+1}^{a+r}, & a<b \\ 
0, & a=b \\ 
-\sum_{r=0}^{a-b-1}x_{i}^{a-r}x_{i+1}^{b+r}, & a>b%
\end{array}%
\right. \;.  \label{nilHeckerep}
\end{equation}%
The difference operators yield a representation the nil Hecke algebra $%
\mathcal{H}_{m}(0)$, that is they obey the relations%
\begin{equation}
\partial _{i}^{2}=-\partial _{i}\qquad \text{and}\qquad \partial
_{i}\partial _{i+1}\partial _{i}=\partial _{i+1}\partial _{i}\partial
_{i+1}\;.  \label{nilHecke}
\end{equation}

\begin{remark}
The action (\ref{nilHeckerep}) of the difference operators $\partial _{i}$
is the familiar action of the Hecke algebra $\mathcal{H}_{m}(\mathfrak{q})$
on the ring of polynomials via Demazure or Bernstein-Gelfand-Gelfand
operators in the limit $\mathfrak{q}\rightarrow 0$. In this limit the
algebra $\mathcal{H}_{m}(0)$ is known as the nil-Hecke algebra.
\end{remark}

We now have the following simple but important lemma.

\begin{lemma}
\label{Lem:ybaDemazure}Let $f\in \mathcal{R}_{m}\otimes V^{\otimes N}$ and
suppose $s_{i}f=f$ for all $i=1,\ldots ,m-1$. Then we have the commutation
relations%
\begin{eqnarray}
A(x_{i+1})B(x_{i})f &=&\partial _{i}B(x_{i+1})A(x_{i})f  \notag \\
D(x_{i+1})B(x_{i})f &=&-\partial _{i}B(x_{i+1})A(x_{i})f  \notag \\
C(x_{i+1})B(x_{i})f &=&\partial _{i}D(x_{i+1})A(x_{i})f\;.
\label{ybaDemazure}
\end{eqnarray}
\end{lemma}

\begin{proof}
This is a direct computation using (\ref{YBE1}), (\ref{R}) and the
definition (\ref{T}).
\end{proof}


To describe the action of the matrix elements $A,B,C,D$ in combinatorial
terms we now relate them to a particular representation of the \emph{affine
nil Temperley-Lieb algebra}. The latter is the unital, associative algebra
generated by $\{u_{1},\ldots ,u_{N}\}$ and relations%
\begin{eqnarray}
u_{i}^{2} &=&u_{i}u_{i+1}u_{i}=u_{i+1}u_{i}u_{i+1}=0,\quad i\in \mathbb{Z}%
_{N}  \notag \\
u_{i}u_{j} &=&u_{j}u_{i},\qquad |i-j|\mod N>1\;.  \label{affTLdef}
\end{eqnarray}%
We will refer to the subalgebra generated by $\{u_{1},\ldots ,u_{N-1}\}$ as
the \emph{finite} nil Temperley-Lieb algebra. Note that the latter is a
quotient of the nil Hecke algebra.

\begin{proposition}[hopping operators]
The map 
\begin{equation}
u_{i}\mapsto f_{i}:=\sigma _{i+1}^{+}\sigma _{i}^{-},\quad \;i=1,\ldots ,N-1
\label{TLrep}
\end{equation}%
yields a faithful representation of the finite nil Temperley-Lieb algebra
over $V_{n}$. If we further set%
\begin{equation}
u_{N}\mapsto f_{N}:=q\sigma _{1}^{+}\sigma _{N}^{-}  \label{affTLrep}
\end{equation}%
we obtain a faithful representation of the \emph{affine} nil Temperley-Lieb
algebra over $\mathbb{C}[q,q^{-1}]\otimes V_{n}$.
\end{proposition}

\begin{remark}
In \cite{KS} a free fermion description of the small quantum cohomology ring
was presented. The relationship between the current description of the
affine nil Temperley-Lieb algebra and that in \emph{loc. cit.} is given via
the following formulae 
\begin{equation*}
\psi _{i}=\sigma _{i}^{-}\prod_{j<i}\sigma _{j}^{z}\qquad \text{and\qquad }%
\psi _{i}^{\ast }=\sigma _{i}^{+}\prod_{j<i}\sigma _{j}^{z}
\end{equation*}%
In particular, one easily verifies that $f_{i}=\psi _{i+1}^{\ast }\psi
_{i}=\sigma _{i+1}^{+}\sigma _{i}^{-}$ with%
\begin{equation*}
\psi _{i+N}=-q^{-1}\psi _{i}\prod_{j=1}^{N}\sigma _{j}^{z}\qquad \text{%
and\qquad }\psi _{i+N}^{\ast }=-q\psi _{i}^{\ast }\prod_{j=1}^{N}\sigma
_{j}^{z}\;.
\end{equation*}%
The last proposition is then an obvious reformulation of \cite[Prop 9.1]{KS}
and we therefore omit the proof.
\end{remark}

The action (\ref{affTLrep}) suggest to introduce the following
quasi-periodic boundary conditions, $\sigma _{i+N}^{\pm }=-q^{\pm 1}\sigma
_{i}^{\pm }$ and $\sigma _{i+N}^{z}=\sigma _{i}^{z}$. We also introduce the
adjoint endomorphisms $f_{i}^{\ast }=\sigma _{i+1}^{-}\sigma _{i}^{+}$ and $%
f_{N}^{\ast }=q^{\ast }\sigma _{1}^{-}\sigma _{N}^{+}$, where $q^{\ast
}=q^{-1}$. For ease of notation, we will henceforth simply write $V_{n,q}:=%
\mathbb{C}[q,q^{-1}]\otimes V_{n}$ and $V_{q}^{\otimes N}:=\mathbb{C}%
[q,q^{-1}]\otimes V^{\otimes N}$.

One easily deduces the following identities which we state without proof;
compare with \cite[Section 8.2 and Lemma 9.3]{KS}.

\begin{lemma}
Denote by $\Theta :b\mapsto b^{\prime },\;\vee :b\mapsto b^{\vee }$ and $%
\func{Rot}:b\mapsto \func{Rot}b$ the endomorphisms $V^{\otimes N}\rightarrow
V^{\otimes N}$ induced by the maps in (\ref{01bijections}) and (\ref{Rotdef}%
). Then%
\begin{equation}
\vee \circ f_{i}=f_{N-i}^{\ast }\circ \vee ,\quad \Theta \circ
f_{i}=f_{N-i}\circ \Theta ,\quad \func{Rot}\circ f_{i+1}=f_{i}\circ \func{Rot%
},  \label{ftrans}
\end{equation}%
where all indices are understood modulo $N$.
\end{lemma}


\begin{proposition}
We have the following expressions for the Yang-Baxter algebra in terms of
the $f_{i}$'s:%
\begin{gather}
A(x)=(1+xf_{N-1})\cdots (1+xf_{1}),  \label{YBA1} \\
B(x)=xA(x)\sigma _{1}^{+},\quad C(x)=\sigma _{N}^{-}A(x),\quad D(x)=x~\sigma
_{N}^{-}A(x)\sigma _{1}^{+}\;.  \notag
\end{gather}
\end{proposition}

\begin{proof}
From the definition of the monodromy matrix one easily derives the expression%
\begin{equation}
T_{b,a}(x)=\sum_{\alpha }x^{|\alpha |+a}(\sigma _{N}^{-})^{b}f_{N-1}^{\alpha
_{N-1}}\cdots f_{1}^{\alpha _{1}}(\sigma _{1}^{+})^{a},  \label{T_exp}
\end{equation}%
where the sum runs over all compositions $\alpha=(\alpha_1,\ldots,%
\alpha_{N-1})$ with $\alpha_i=0,1$. The assertion is now immediate.
\end{proof}


The action of the polynomials $A_{r}$ in $V_{n}\subset V^{\otimes N}$ is
easily described using the well-known bijections between 01-words and
partitions explained in Section 2.

\begin{lemma}[horizontal strips]
Let $\mu \in (n,k),$ and $A(x)=\sum_{r\geq 0}x^{r}A_{r}$. Then the
polynomials%
\begin{equation}
A_{r}=\sum_{\alpha \vdash r}f_{N-1}^{\alpha _{N-1}}\cdots f_{1}^{\alpha _{1}}
\label{Ar}
\end{equation}%
act on the basis vector $|\mu \rangle $ by adding all possible horizontal $r$%
-strips to the Young diagram of $\mu $ such that the result $\lambda $ lies
within the $n\times k$ bounding box, $A_{r}|\mu \rangle =\sum_{\lambda -\mu
=(r)}|\lambda \rangle $.
\end{lemma}

\begin{proof}
Using the bijection (\ref{part2word}) one readily verifies that either $%
f_{i}|\mu \rangle =|\lambda \rangle $, where $\lambda $ is obtained by
adding a box in the $(i-n)$th diagonal of the Young diagram of $\mu $, or,
if this is not possible $f_{i}|\mu \rangle =0$. Consider now a consecutive
string $f_{i+r^{\prime }}\cdots f_{i+1}f_{i}|\mu \rangle =|\lambda \rangle $
with $r^{\prime }\leq r$ and suppose $w_{i}(\mu )=1$, $w_{j}(\mu )=0$ for $%
i<j\leq i+r^{\prime }$; otherwise the action is trivial. Then the 1-letter
at position $i$ in $w(\mu )$ is moved past $r^{\prime }$ 0-letters whose
position each decreases by one. Since $\mu _{k+1-j}^{\prime }=\ell _{j}(\mu
^{\prime })+j$, where $N+1-\ell _{k+1-j}(\mu ^{\prime })$ are the positions
of 0-letters in $w(\mu ),$ we find $\lambda _{k+1-j}^{\prime }-\mu
_{k+1-j}^{\prime }=1$. Thus, $\lambda $ is obtained from $\mu $ by adding a
horizontal strip of length $r^{\prime }$. This proves the assertion.
\end{proof}

From the last lemma the action of the remaining Yang-Baxter algebra
generators is obtained by observing that $\sigma _{1}^{+}|\lambda \rangle
=|\lambda _{1}-1,\ldots ,\lambda _{n}-1\rangle $ if $\lambda _{n}>0$ or $%
\sigma _{1}^{+}|\lambda \rangle =0$ if $\lambda _{n}=0$. In contrast, $%
\sigma _{N}^{-}|\lambda \rangle =|\mu \rangle $ if $\mu $ can be obtained by
adding a column of maximal height to the Young diagram of $\lambda $ and
then subtracting a boundary ribbon of length $N$ starting in the first row.
Otherwise, we have $\sigma _{N}^{-}|\lambda \rangle =0$.

\begin{lemma}
(i) The operator $B_{r}$ acts by adding all possible horizontal strips of
length $r-1$ such that the result lies within the $n\times k$ bounding box
and removing the leftmost column of the resulting Young diagram afterwards.
(ii) In contrast, the operator $C_{r}$ acts be adding all possible
horizontal strips of length $r$ and removing the top row afterwards. (iii)
Finally, the action of the operator $D_{r}$ is to first add an horizontal $r$%
-strip, preserving the height but not the width of the bounding box, and
then subtracting a boundary ribbon of length $N$ starting from the first row.
\end{lemma}


From the commutation relations (\ref{YBA1a}) of the Yang-Baxter algebra it
now follows that certain polynomials in the noncommutative alphabets $%
\{f_{1},\ldots ,f_{N-1}\}$ and $\{f_{1},\ldots ,f_{N}\}$ commute.

\begin{corollary}[integrability]
(i) The finite nil Temperley-Lieb polynomials $A_{r}$ commute pairwise, that
is $A_{r}A_{r^{\prime }}=A_{r^{\prime }}A_{r},\ \forall r,r^{\prime
}=0,1,\ldots ,N$.\newline
(ii) Let $H(x)=A(x)+qD(x)\in \func{End}V_{q}^{\otimes N}$ then one has the
expansion%
\begin{equation}
H(x)=\sum_{r=0}^{N}x^{r}H_{r},\quad H_{r}:=\sum_{\alpha \vdash r}q^{\alpha
_{N}}(\sigma _{N}^{-})^{\alpha _{N}}f_{N-1}^{\alpha _{N-1}}\cdots
f_{1}^{\alpha _{1}}(\sigma _{1}^{+})^{\alpha _{N}}  \label{H}
\end{equation}%
%
%
%
%
%
%
%
%
%
%
%
%
%
%
%
%
%
%
%
%
%
%
%
%
%
%
%
%
%
%
%
%
%
%
%
%
%
%
Moreover, $H_{r}H_{r^{\prime }}=H_{r^{\prime }}H_{r}$ for all $r,r^{\prime
}=0,1,\ldots ,N$.
\end{corollary}

\begin{proof}
The first assertion is immediate from the commutation relation (\ref{YBA1a}%
). The operator $H:V_{q}^{\otimes N}\rightarrow V_{q}^{\otimes N}$ is the
so-called row-to-row transfer matrix of the vicious walker model on the
cylinder, i.e. its matrix elements are the partition functions of one
lattice row when imposing quasi-periodic boundary conditions in the
horizontal direction of the square lattice. The transfer matrix can be
written as the following partial trace, 
\begin{equation}
H(x)=\func{Tr}_{0}q^{(\sigma ^{+}\sigma ^{-})_{0}}L_{0N}(x)\cdots
L_{01}(x)\,,  \label{5vH2}
\end{equation}%
where the indices indicate in which factors of the tensor product $%
V(x)\otimes V_{q}^{\otimes N}$ the respective $L$-operators act. (The
factors are labeled from left to right starting as $0,1,\ldots ,N$.) The
additional operator $q^{\sigma ^{+}\sigma ^{-}}$ under the trace invokes
quasi-periodic boundary conditions, i.e. the powers of the indeterminate $q$
count how many outer horizontal edges are occupied. It is now a consequence
of the Yang-Baxter equation that this model is integrable, i.e. the transfer
matrices commute pairwise, $H(x)H(y)=H(y)H(x)$ for any $x,y$. The assertion
then follows from the expansion (\ref{T_exp}).
\end{proof}

For completeness we summarise our previous findings on the combinatorial
action of the Yang-Baxter algebra (\ref{YBA1}) in the following formula for
each $r=0,1,\ldots ,N$.

\begin{lemma}
Let $\mu \in (n,k)$ then%
\begin{equation}
H_{r}|\mu \rangle =\sum_{\lambda -\mu =(r)}|\lambda \rangle +q\sum_{\lambda
\lbrack 1]-\mu =(r)}|\lambda \rangle ,  \label{Haction}
\end{equation}%
where in the second sum $\lambda \lbrack 1]$ denotes the partition obtained
from $\lambda $ by adding a boundary ribbon of length $N$ starting in the
first and ending in the $n^{\text{th}}$ row. For $r>k$ we have $H_{r}|\mu
\rangle =0$.
\end{lemma}

\begin{proof}
We postpone a detailed proof to Section \ref{sec:bijections} where we
discuss the bijection between row configurations of the vicious walker model
and toric tableaux; see the proof of Prop \ref{prop:bijection}.
\end{proof}

\subsection{Osculating walkers: vertex and lattice configurations}

Define another 5-vertex model but this time on a $k\times N$ lattice with $%
k=N-n$,%
\begin{equation}
\mathbb{L}^{\prime }:=\{\langle i,j\rangle \in \mathbb{Z}^{2}|0\leq i\leq
k+1,\;0\leq j\leq N+1\}\,.
\end{equation}%
Denote by $\mathbb{E}^{\prime }$ the set of its horizontal and vertical
edges. As before we define the weight of a lattice configuration $\mathcal{C}%
:\mathbb{E}^{\prime }\rightarrow \{0,1\}$ as $\func{wt}\nolimits^{\prime }(%
\mathcal{C})=\prod_{(i,j)\in \mathbb{L}^{\prime }}\func{wt}\nolimits^{\prime
}(\mathrm{v}_{i,j})\in \mathbb{Z}[x_{1},\ldots ,x_{k}]$, where the allowed
vertex configurations and their weights $\func{wt}\nolimits^{\prime }(%
\mathrm{v}_{i,j})$ are depicted in Figure \ref{fig:QCvertex2}. 
\begin{figure}[tbp]
\begin{equation*}
\includegraphics[scale=0.35]{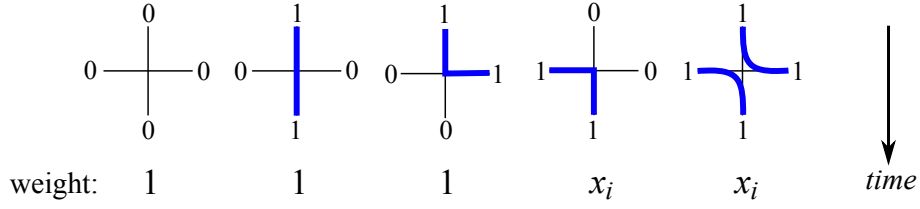}
\end{equation*}%
\caption{The five allowed vertex configurations for the osculating walker
model. Note that this model differs from the vicious walker model in the
last vertex configuration only.}
\label{fig:QCvertex2}
\end{figure}
The weights of the allowed vertex configurations again depend only on the
row index of the square lattice via the commutative indeterminates $%
(x_{1},\ldots ,x_{k})$. The corresponding $L^{\prime }$-matrix reads%
\begin{equation}
L^{\prime }v_{a}\otimes v_{b}=\sum_{c,d=0,1}L_{cd}^{\prime \,ab}v_{c}\otimes
v_{d}=x_{i}^{a}[v_{0}\otimes (\sigma ^{+})^{a}v_{b}+v_{1}\otimes (\sigma
^{+})^{a}\sigma ^{-}v_{b}]\;,  \label{L'matrix}
\end{equation}%
which can be rewritten in the basis $\{v_{0}\otimes v_{0},v_{0}\otimes
v_{1},v_{1}\otimes v_{0},v_{1}\otimes v_{1}\}$ as 
\begin{equation*}
L^{\prime }(x_{i})=\left( 
\begin{array}{cccc}
1 & 0 & 0 & 0 \\ 
0 & 1 & x_{i} & 0 \\ 
0 & 1 & 0 & 0 \\ 
0 & 0 & 0 & x_{i}%
\end{array}%
\right) \;.
\end{equation*}%
Via a straightforward computation, which we omit, one arrives at the
following result.

\begin{proposition}
We have the identity 
\begin{equation}
R_{12}^{\prime }(x/y)L_{13}^{\prime }(x)L_{23}^{\prime }(y)=L_{23}^{\prime
}(y)L_{13}^{\prime }(x)R_{12}^{\prime }(x/y),  \label{YBE2}
\end{equation}%
where 
\begin{equation}
R^{\prime }(x/y)=\left( 
\begin{array}{cccc}
1 & 0 & 0 & 0 \\ 
0 & 1-x/y & x/y & 0 \\ 
0 & 1 & 0 & 0 \\ 
0 & 0 & 0 & x/y%
\end{array}%
\right)\;.  \label{R'}
\end{equation}
\end{proposition}

Note that the matrix $R^{\prime}(x/y)$ is non-singular for generic $x,y$.

In complete analogy with our previous discussion of the vicious walker
model, we can define also here a monodromy matrix 
\begin{equation}
T^{\prime }(x)=L_{0N}^{\prime }(x)\cdots L_{02}^{\prime }(x)L_{01}^{\prime
}(x)=\left( 
\begin{array}{cc}
A^{\prime }(x) & B^{\prime }(x) \\ 
C^{\prime }(x) & D^{\prime }(x)%
\end{array}%
\right)  \label{5vT'}
\end{equation}%
and Yang-Baxter algebra.

\begin{proposition}
The matrix elements in (\ref{5vT'}) are given by the following expressions
in the hopping operators (\ref{TLrep}),%
\begin{gather}
A^{\prime }(x)=(1+xf_{1})\cdots (1+xf_{N-1})  \label{A'} \\
B^{\prime }(x)=x\sigma _{1}^{+}A^{\prime }(x),\quad C^{\prime }(x)=A^{\prime
}(x)\sigma _{N}^{-},\quad D^{\prime }(x)=x~\sigma _{1}^{+}A^{\prime
}(x)\sigma _{N}^{-}  \label{YBA2}
\end{gather}%
and we have the commutation relations $A^{\prime }(x)A^{\prime
}(y)=A^{\prime }(y)A^{\prime }(x)$.
\end{proposition}

\begin{proof}
Exploiting (\ref{L'matrix}) one derives from the definition (\ref{5vT'}) the
expansion%
\begin{equation}
T_{b,a}^{\prime }(x)=\sum_{\alpha }x^{|\alpha |+a}(\sigma
_{1}^{+})^{a}f_{1}^{\alpha _{1}}\cdots f_{N-1}^{\alpha _{N-1}}(\sigma
_{N}^{-})^{b}\;,  \label{T'_exp}
\end{equation}
where the sum runs again over all compositions $\alpha=(\alpha_1,\ldots,%
\alpha_{N-1})$ with $\alpha_i=0,1$.
\end{proof}

Similar like before one now verifies the following combinatorial action of
the Yang-Baxter algebra.

\begin{lemma}[vertical strips]
Let $A^{\prime }(x)=\sum_{r\geq 0}x^{r}A_{r}^{\prime }$. The polynomials 
\begin{equation}
A_{r}^{\prime }=\sum_{\alpha \vdash r}f_{1}^{\alpha _{1}}\cdots
f_{N-1}^{\alpha _{N-1}}  \label{A'r}
\end{equation}%
act on $|\mu \rangle $ with $\mu \in (n,k)$ a partition by adding all
possible vertical $r$-strips to the Young diagram of $\mu $ such that the
resulting diagram lies still within the $n\times k$ bounding box, $%
A_{r}^{\prime }|\mu \rangle =\sum_{\lambda /\mu =(1^{r})}|\lambda \rangle $.
\end{lemma}

\begin{proof}
Similar to the case of vicious walkers we consider $f_{i}f_{i+1}\cdots
f_{i+r}|\mu \rangle =|\lambda \rangle $ which is non-trivial only if $%
w_{i+r^{\prime }+1}(\mu )=0$ and $w_{j}(\mu )=1$ with $i\leq j\leq
i+r^{\prime },\;r^{\prime }\leq r$. Employing once more (\ref{part2word}) we
find that $\lambda _{n+1-j}-\mu _{n+1-j}=1$. Thus, $\lambda $ is obtained by
adding a vertical strip of height $r$.
\end{proof}

We now impose again quasi-periodic boundary conditions in the horizontal
direction of the square lattice and introduce the corresponding transfer
matrix $E:V_{q}^{\otimes N}\rightarrow V_{q}^{\otimes N}$,%
\begin{equation}
E(x_{i}):=A^{\prime }(x_{i})+qD^{\prime }(x_{i})=\func{Tr}_{0}q^{(\sigma
^{+}\sigma ^{-})_{0}}L_{0N}^{\prime }(x_{i})\cdots L_{01}^{\prime }(x_{i})\;.
\label{5vE}
\end{equation}%
As in the case of the vicious walkers model one now exploits the Yang-Baxter
equation to arrive at the following set of commuting polynomials in the $%
f_{i}$'s.

\begin{corollary}[integrability]
We have the expansion 
\begin{equation}
E(x)=\sum_{r=0}^{N}x^{r}E_{r},\quad E_{r}=\sum_{\alpha \vdash r}q^{\alpha
_{N}}(\sigma _{1}^{+})^{\alpha _{N}}f_{1}^{\alpha _{1}}\cdots
f_{N-1}^{\alpha _{N-1}}(\sigma _{N}^{-})^{\alpha _{N}}  \label{E}
\end{equation}%
and the $\{E_{r}\}_{r=0}^{N}$ commute pairwise.
\end{corollary}

The following lemma shows that vicious and osculating walkers are related
via level-rank duality.

\begin{lemma}[level-rank duality]
\label{levelrank} Let $\Theta :V_{q}^{\otimes N}\rightarrow V_{q}^{\otimes
N} $ be the involution induced by sending each $|w\rangle \in \mathcal{B}$
to $|w^{\prime }\rangle $; compare with (\ref{01bijections}). Then we have $%
\Theta \circ H_{r}=E_{r}\circ \Theta $ for all $r=0,1,\ldots ,N$.
\end{lemma}

\begin{proof}
This is an easy consequence of the second relation in (\ref{ftrans}) and the
formulae (\ref{H}), (\ref{E}).
\end{proof}

An immediate consequence is the following combinatorial action of the $E_{r}$%
's which simply follows from the combinatorial action of the vicious walker
transfer matrix discussed previously.

\begin{lemma}
Let $\mu \in (n,k)$ then%
\begin{equation}
E_{r}|\mu \rangle =\sum_{\lambda -\mu =(1^{r})}|\lambda \rangle
+q\sum_{\lambda \lbrack 1]-\mu =(1^{r})}|\lambda \rangle ,  \label{Eaction}
\end{equation}%
where in the second sum $\lambda \lbrack 1]$ again denotes the partition
obtained from $\lambda $ by adding a boundary ribbon of length $n+k$
starting in the first and ending in the $n^{\text{th}}$ row. For $r>n$ we
have $E_{r}|\mu \rangle =0$.
\end{lemma}

We can also determine the commutation relations between the Yang-Baxter
algebras of the vicious and osculating walker models via a third and final
Yang-Baxter relation, which - again - is obtained by a tedious but direct
computation which we omit.

\begin{proposition}
We have the additional identity 
\begin{equation}
R_{12}^{\prime \prime }(x/y)L_{1}(x)L_{2}^{\prime }(y)=L_{2}^{\prime
}(y)L_{1}(x)R_{12}^{\prime \prime }(x/y)  \label{YBE3}
\end{equation}%
where $R^{\prime \prime }$ is the singular matrix%
\begin{equation}
R^{\prime \prime }(x/y)=\left( 
\begin{array}{cccc}
1+x/y & 0 & 0 & 0 \\ 
0 & 1 & x/y & 0 \\ 
0 & 1 & x/y & 0 \\ 
0 & 0 & 0 & 0%
\end{array}%
\right) \;.  \label{R''}
\end{equation}
\end{proposition}

Exploiting this last result one now proves in a similar manner as before
that the transfer matrices $H(x)$ and $E(y)$ commute for arbitrary $x,y$.

\begin{corollary}
We have the following commutation relations%
\begin{eqnarray}
A(x)A^{\prime }(y) &=&A^{\prime }(y)A(x),  \notag \\
(x+y)A^{\prime }(y)\sigma _{N}^{-}A(x) &=&x\sigma _{N}^{-}A(x)A^{\prime
}(y)+yA(x)A^{\prime }(y)\sigma _{N}^{-},  \notag \\
(x+y)A(x)\sigma _{1}^{+}A^{\prime }(y) &=&xA(x)A^{\prime }(y)\sigma
_{1}^{+}+y\sigma _{1}^{+}A(x)A^{\prime }(y),  \label{YBA3}
\end{eqnarray}%
and%
\begin{equation}
yA(x)\sigma _{1}^{+}A^{\prime }(y)\sigma _{N}^{-}+x\sigma _{N}^{-}A(x)\sigma
_{1}^{+}A^{\prime }(y)= y\sigma _{1}^{+}A^{\prime }(y)\sigma
_{N}^{-}A(x)+xA^{\prime }(y)\sigma _{N}^{-}A(x)\sigma _{1}^{+}\ .
\end{equation}%
In particular, finite and affine nil Temperley-Lieb polynomials $H_{r},E_{r}$
pairwise commute, i.e. $H_{r}E_{r^{\prime }}=E_{r^{\prime }}H_{r}$ for all $%
r,r^{\prime }=0,1,\ldots ,N$.
\end{corollary}

We have the following functional relation between the transfer matrices $H,E$
which generalises the known relation of the generating functions for
elementary and complete symmetric polynomials in the ring of symmetric
functions.

\begin{proposition}
\label{TQ} We have the operator identity%
\begin{equation}
H(x)E(-x)=1+qx^{N}\prod_{j=1}^{N}\sigma _{j}^{z}\;.  \label{TQ_eqn}
\end{equation}
\end{proposition}

\begin{proof}
Decompose $V\otimes V$ into $W\oplus W^{\perp }$ with $W=\text{ker }%
R^{\prime \prime }(-1)$. Namely, $W$ is spanned by $w_{1}=v_{0}\otimes v_{0}$%
, $w_{2}=v_{0}\otimes v_{1}+v_{1}\otimes v_{0}$, $w_{3}=v_{1}\otimes v_{1}$,
while $W^{\perp }=\mathbb{C}\{w_{4}\}$ with $w_{4}=\frac{1}{2}(v_{1}\otimes
v_{0}-v_{0}\otimes v_{1})$. From the Yang-Baxter equation (\ref{YBE3}) for $%
y=-x$ one deduces that $L_{13}(x)L_{23}^{\prime }(-x)W\otimes V\subset
W\otimes V$. Thus, we can block decompose $L_{13}(x)L_{23}^{\prime
}(-x)=\left( 
\begin{smallmatrix}
M & \ast \\ 
0 & M^{\prime }%
\end{smallmatrix}%
\right) $ with respect to $W\oplus W^{\perp }$ and one only needs to verify
that $M,M^{\prime }$ yield the asserted terms on the right hand side of (\ref%
{TQ_eqn}) using that 
\begin{equation*}
L_{13}(x)L_{23}^{\prime }(-x)w_{1}\otimes v_{0}=w_{1}\otimes v_{0},\quad
L_{13}(x)L_{23}^{\prime }(-x)w_{1}\otimes v_{1}=w_{1}\otimes
v_{1}+w_{2}\otimes v_{0}
\end{equation*}%
and 
\begin{equation*}
L_{13}(x)L_{23}^{\prime }(-x)w_{2}\otimes v_{0,1}=L_{13}(x)L_{23}^{\prime
}(-x)w_{3}\otimes v_{0,1}=0\;.
\end{equation*}%
We leave the details of this last step to the reader since it is a simple
computation. 
\end{proof}

\begin{corollary}
\label{cor:iso}Let $\mathcal{A}_{n,k}\subset \limfunc{End}(V_{n,q})$ be the
commutative algebra generated by $\{H_{j}\}_{j=0}^{k}$ and $%
\{E_{i}\}_{i=0}^{n}$. The map $E_{i}\mapsto e_{i}$ and $H_{j}\mapsto h_{j}$
provides a canonical algebra isomorphism $\mathcal{A}_{n,k}\cong qH^{\ast }(%
\limfunc{Gr}_{n,n+k})\otimes _{\mathbb{Z}}\mathbb{C}$.

\begin{proof}
This is clear from our previous results: we have established that the
combinatorial actions (\ref{Haction}) and (\ref{Eaction}). In particular,
one has that $H_{r}|V_{n,q}=0$ for $r>k$ and $E_{r}|V_{n,q}=0$ for $r>n$.
The Yang-Baxter relations (\ref{YBE1}), (\ref{YBE2}), (\ref{YBE3}) provided
us with a proof that the $E_{i}$'s and $H_{j}$'s commute among themselves
and with each other. Finally, the last result (\ref{TQ_eqn}) gives the
desired algebraic dependence expressed in (\ref{Ipoly0}).
\end{proof}
\end{corollary}

\section{Algebraic Bethe ansatz and idempotents}

We show in this section that the eigenvectors of the transfer matrices $H,E$
for the vicious and osculating walker models are the idempotents of the
fusion ring of the gauged WZNW model. The eigenbasis of the affine nil
Temperley-Lieb polynomials $H_{\lambda }:=H_{\lambda _{1}}H_{\lambda
_{2}}\ldots $ and $E_{\lambda }:=E_{\lambda _{1}}E_{\lambda _{2}}\ldots $
with $\lambda _{i}<N$ has been previously constructed in \cite[Section 10]%
{KS} using the free fermion formalism mentioned earlier.

Here we obtain a \emph{new} result: we show that the same eigenbasis can be
obtained from the Yang-Baxter algebras (\ref{YBA1}), (\ref{YBA2}) of the
vicious and osculating walker models by a procedure known as \emph{algebraic
Bethe ansatz}; see e.g. \cite{KIB} for a textbook reference. In partiuclar,
this construction will furnish us with the vertex-type operator formulae (%
\ref{Bethev1}), (\ref{Bethev2}) for Schur functions which differ from the
known expressions; see e.g. \cite[Chapter I, Section 5, Ex 29]{Macdonald}
and references therein.

\subsection{Bethe vectors}

We assume that $q^{\pm 1/N}$ exist and define $\mathfrak{e}_{\lambda }=%
\mathfrak{e}(y_{1}(\lambda ),\ldots ,y_{n}(\lambda ))$ where%
\begin{equation}
\mathfrak{e}(y_{1},\ldots ,y_{n}):=\sum_{\lambda \in (n,k)}s_{\lambda
}(y_{1}^{-1},\ldots ,y_{n}^{-1})|\lambda \rangle  \label{Bethevector}
\end{equation}%
and $y_{j}(\lambda )=q^{\frac{1}{N}}\exp [\frac{2\pi i}{N}J_{j}(\lambda )]$
with $J_{j}(\lambda )=-\tfrac{n+1}{2}+\lambda _{j}+j$ and $\lambda \in (n,k)$%
. We recall from \cite[Section 10]{KS} the following results.

\begin{theorem}[Korff-Stroppel]
\label{completeness} (i) The vectors $\{\mathfrak{e}_{\lambda }\}_{\lambda
\in (n,k)}$ form an orthogonal basis of $V_{n}$ and we have that%
\begin{equation}
||\mathfrak{e}_{\lambda }||^{2}:=\sum_{\mu }s_{\mu }(y_{1},\ldots
,y_{n})s_{\mu }(y_{1}^{-1},\ldots ,y_{n}^{-1})=\frac{2^{\frac{n(1-n)}{2}%
}k(n+k)^{k}}{\prod_{i<j}\sin ^{2}\frac{\pi }{n+k}(\lambda _{i}-\lambda
_{j}+i-j)}\;.
\end{equation}%
(ii) The basis $\{\mathfrak{e}_{\lambda }\}_{\lambda \in (n,k)}$
diagonalises the $H_{\mu }$'s and $E_{\mu }$'s; one has $H_{\mu }\mathfrak{e}%
_{\lambda }=h_{\mu }(y(\lambda ))\mathfrak{e}_{\lambda }$ and $E_{\mu }%
\mathfrak{e}_{\lambda }=e_{\mu }(y(\lambda ))\mathfrak{e}_{\lambda }$ for
all compositions $\mu $ with $\mu _{i}<N$.
\end{theorem}

The following identities which connect the eigenvectors with the Yang-Baxter
algebras are new and are not contained in \cite[Section 10]{KS}.

\begin{lemma}
\label{lem:Bethe} Let $|0\rangle $, $|N\rangle $ be the basis vectors in the
one-dimensional spaces $V_{0}$ and $V_{N}$. We have the following expansions,%
\begin{equation}
B(x_{1})\cdots B(x_{n})|0\rangle =x^{\delta _{n}}\sum_{\lambda \in
(n,k)}s_{\lambda }(x)|\lambda \rangle  \label{Bethev1}
\end{equation}%
with $\delta _{n}=(n,n-1,\ldots ,1)$ and%
\begin{equation}
C^{\prime }(x_{1})\cdots C^{\prime }(x_{k})|N\rangle =x^{\rho
_{k}}\sum_{\lambda \in (n,k)}s_{\lambda ^{\prime }}(x)|\lambda \rangle \;,
\label{Bethev2}
\end{equation}%
where $\rho _{k}=(k-1,\ldots ,2,1)$.
\end{lemma}

So in particular we have the following identities between matrix elements
and Schur functions, 
\begin{equation}
s_{\lambda }(x_{1},\ldots ,x_{n})=x^{-\delta }\langle \lambda
|B(x_{1})\cdots B(x_{n})|0\rangle =x^{-\rho _{n}}\langle \lambda ^{\prime
}|C^{\prime }(x_{1})\cdots C^{\prime }(x_{n})|N\rangle \;.
\end{equation}

\begin{proof}
The proof is graphical. Draw a diagonal line across the square lattice as
indicated in Figure \ref{fig:Bop}.

\begin{claim}
\label{claim} The vicious walker configurations shown on the right in Figure %
\ref{fig:Bop} are in bijection with semi-standard tableaux $T$ of shape $%
\lambda \in (n,k)$ where $\ell _{i}(\lambda ),\,i=1,\ldots n$ defined in (%
\ref{part2word}) are the end positions of the walkers and the corresponding
weight of such configurations is $x^{\alpha }$, where $\alpha \in \mathbb{Z}%
_{\geq 0}^{n}$ is the weight of the tableau.
\end{claim}

We postpone the proof of this claim to the next section where we establish a
bijection between non-intersecting paths on the cylinder and toric Young
tableaux. The claim will then follow as a special case by considering only
those vertex weights which lie below the dotted line and noting that the
vertices above the line contribute the total weight factor $%
x_{1}^{n}x_{2}^{n-1}\cdots x_{n}$. Thus, redrawing the lattice paths as
indicated in Figure \ref{fig:Bop} and using the known sum formula $%
s_{\lambda }(x)=\sum_{|T|=\lambda }x^{T}$ (see e.g. \cite[Chapter I]%
{Macdonald}) we arrive at the desired expression 
\begin{equation*}
B(x_{n})\cdots B(x_{1})|0\rangle =x_{1}^{n}x_{2}^{n-1}\cdots
x_{n}\sum_{\lambda \in (n,k)}s_{\lambda }(x)|\lambda \rangle \;.
\end{equation*}

\begin{figure}[tbp]
\begin{equation*}
\includegraphics[scale=0.45]{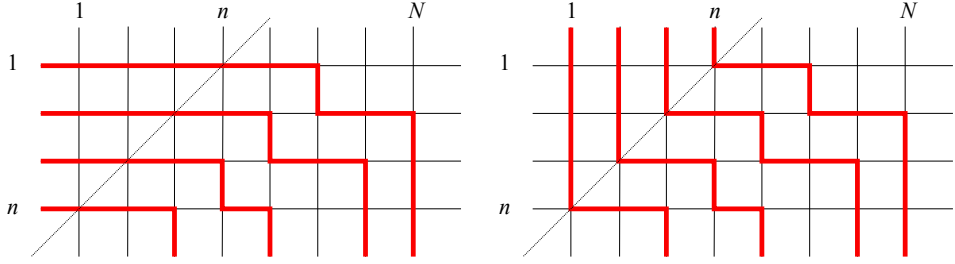}
\end{equation*}%
\caption{Vicious walker configurations for the $B$-operator (left) and the $%
A $-operator (right).}
\label{fig:Bop}
\end{figure}

The second identity simply uses the fact that $\Theta \sigma _{i}^{+}\Theta
=\sigma _{N-i}^{-}$ which is easily verified. Thus, from $\Theta A(x)\Theta
=A^{\prime }(x)$ (see Lemma \ref{levelrank} for $q=0$) and (\ref{YBA1}), (%
\ref{YBA2}) it follows that $\Theta B(x)\Theta =xC^{\prime }(x)$. Applying
the involution $\Theta $ on both sides of the previous identity and %
swapping $n$ and $k$ afterwards, the second assertion is proved.
\end{proof}

\begin{proposition}
Fix $0\leq n<N$. For each $\lambda \in (n,k)$ the vector $\mathfrak{e}%
_{\lambda }$ is an eigenvector of $H=A+qD$ and $E=A^{\prime }+qD^{\prime }$
with%
\begin{eqnarray}
H(x)\mathfrak{e}_{\lambda } &=&(1+(-1)^{n}qx^{N})\prod_{i=1}^{n}\frac{1}{%
1-xy_{i}(\lambda )}~\mathfrak{e}_{\lambda }\;,  \label{Hspec} \\
E(x)\mathfrak{e}_{\lambda } &=&\prod_{i=1}^{n}(1+xy_{i}(\lambda ))~\mathfrak{%
e}_{\lambda }.  \label{Espec}
\end{eqnarray}%
For $n=N$ the eigenvalue equations simplify to $H(x)|N\rangle =|N\rangle $
and $E(x)|N\rangle =(1+qx^{N})|N\rangle $. 
\end{proposition}

\begin{proof}
Start with $n=1$. Note that $A(x)|0\rangle =|0\rangle $ and $D(x)|0\rangle
=x^{N}|0\rangle $. According to (\ref{YBA1b}) we then find that 
\begin{equation*}
B(y)|0\rangle =\sum_{r>0}y^{r}|0\cdots 0\underset{r}{1}0\cdots 0\rangle
\end{equation*}%
is an eigenvector of $H(x)$ with eigenvalue $(1+qx^{N})/(1-x/y)$ provided
that $y^{N}q=1$. This computation generalizes to $n>1$, using an induction
argument one finds with the help of the equations (\ref{YBA1b}),%
\begin{eqnarray*}
A(x)B(y_{n})\cdots B(y_{1}) &=&B(y_{n})\cdots B(y_{1})A(x)\prod_{i=1}^{n}%
\frac{y_{i}}{y_{i}-x} \\
&&-\sum_{i=1}^{n}B(y_{n})\cdots \underset{i}{B(x)}\cdots B(y_{1})A(y_{i})%
\frac{y_{i}}{y_{i}-x}\prod_{j\neq i}\frac{y_{i}y_{j}}{y_{j}-y_{i}}
\end{eqnarray*}%
%
%
%
%
%
%
%
%
%
%
%
%
%
%
and%
\begin{eqnarray*}
D(x)B(y_{n})\cdots B(y_{1}) &=&B(y_{n})\cdots B(y_{1})D(x)\prod_{i=1}^{n}%
\frac{y_{i}}{x-y_{i}} \\
&&-\sum_{i=1}^{n}B(y_{n})\cdots \underset{i}{B(x)}\cdots B(y_{1})D(y_{i})%
\frac{y_{i}}{x-y_{i}}\prod_{j\neq i}\frac{y_{i}y_{j}}{y_{i}-y_{j}}\;,
\end{eqnarray*}%
%
%
%
%
%
%
%
%
%
%
%
%
%
%
Employing these identities one deduces that for $\mathfrak{e}(y)$ to be an
eigenvector the so-called Bethe roots $y_{i}$ have to satisfy the following
set of constraints,%
\begin{equation}
y_{1}^{N}=\cdots =y_{n}^{N}=(-1)^{n-1}q\;.  \label{bae}
\end{equation}%
The explicit solution to these equations is easily obtained and can be found
in \cite[Prop 10.4]{KS}.

To arrive at the eigenvalue equation for $E$ one can either perform a
similar computation using the operators $C^{\prime }$ introduced earlier and
employing level-rank duality or employ Prop \ref{TQ}.

The statement for $n=N$ is obvious and follows from the definition of $H,E$.
\end{proof}

Let $\mathcal{P}$ be the linear operator $V^{\otimes N}\rightarrow
V^{\otimes N}$ defined by $\mathcal{P}|\lambda \rangle =|\lambda ^{\vee
}\rangle $ and $\mathcal{T}$ the linear operator defined by $\mathcal{T}%
|\lambda \rangle =|\lambda \rangle $ and $\mathcal{T}q=q^{-1}\mathcal{T}$
for each $\lambda \in (n,k)$ with $N=n+k$. We introduce the operators 
\begin{equation}
H^{\ast }(x)=\mathcal{PT}~H(x)~\mathcal{PT}\quad \text{and}\quad E^{\ast
}(x)=\mathcal{PT}~E(x)~\mathcal{PT\;}.  \label{H*&E*}
\end{equation}%
Note that it follows from the definition that we have the identity $H^{\ast
}(x)=A^{\ast }(x)+q^{-1}D^{\ast }(x)$ with%
\begin{equation*}
L_{01}(x)L_{02}(x)\cdots L_{0N}(x)=\left( 
\begin{array}{cc}
A^{\ast }(x) & B^{\ast }(x) \\ 
C^{\ast }(x) & D^{\ast }(x)%
\end{array}%
\right) \;.
\end{equation*}%
An analogous formula holds for $E^{\ast }(x)$.

\begin{lemma}
We have the dual affine Pieri rules%
\begin{equation}
H_{r}^{\ast }|\mu \rangle =\sum_{\mu -\lambda =(r)}|\lambda \rangle
~+q^{-1}\sum_{\mu \lbrack 1]-\lambda =(r)}|\lambda \rangle
\label{dualPieri1}
\end{equation}%
and%
\begin{equation}
E_{r}^{\ast }|\mu \rangle =\sum_{\mu -\lambda =(1^{r})}|\lambda \rangle
~+q^{-1}\sum_{\mu \lbrack 1]-\lambda =(1^{r})}|\lambda \rangle \;,
\label{dualPieri2}
\end{equation}%
where the notation $\mu \lbrack 1]$ in the second sum in both formulae
stands for the partion obtained by adding a boundary rim hook of length $n+k$
to $\mu $.
\end{lemma}

\begin{proof}
Simply note that $\mathcal{P}\sigma _{i}^{+}\sigma _{i+1}^{-}=\sigma
_{N-i}^{-}\sigma _{N+1-i}^{+}\mathcal{P}$. The rest is then a
straightforward computation which follows along similar lines as in the case
of $H_{r}$ and $E_{r}$ using the explicit polynomial expressions in the $%
f_{i}$'s given in (\ref{H}), (\ref{E}).
\end{proof}

The last result motivates us to define an isomorphism $(V_{q}^{\otimes
N})^{\ast }\rightarrow V_{q}^{\otimes N}$ by identifying the dual basis
vector $\langle \lambda |$ with $\mathcal{P}|\lambda \rangle =|\lambda
^{\vee }\rangle $, and more generally any vector $v^{\ast }=\sum_{\lambda
}c_{\lambda }(q)\langle \lambda |~\in (V_{q}^{\otimes N})^{\ast }$ with $%
v=\sum_{\lambda }c_{\lambda }(q^{-1})|\lambda ^{\vee }\rangle \in
V_{q}^{\otimes N}$. In other words the pairing $\langle v|w\rangle $ between
the vector space and its dual is given by the bilinear form $\eta (v,w):=(v,%
\mathcal{PT}w)$, where $(|\lambda \rangle ,|\mu \rangle )=\delta _{\lambda
\mu }$.

\begin{lemma}
The \emph{left} or \emph{dual} eigenvectors $\mathfrak{e}_{\lambda }^{\ast }$
defined by $\langle \mathfrak{e}_{\lambda }^{\ast }|\mathfrak{e}_{\mu
}\rangle =\delta _{\lambda \mu }$ are given by the expansion%
\begin{equation}
\mathfrak{e}_{\lambda }^{\ast }=\sum_{\mu \in (n,k)}\frac{s_{\mu }(y(\lambda
))}{||\mathfrak{e}_{\lambda }||^{2}}\langle \mu |\;.  \label{dualBethe}
\end{equation}
\end{lemma}

\begin{proof}
This is a direct consequence of our previous discussion, applying $\mathcal{%
PT}$ to (\ref{Bethev1}) and then using part (i) of Theorem \ref{completeness}%
.
\end{proof}

\begin{remark}
In \cite{KS} the operator $\hat{H}_{N}=-q\prod_{j=1}^{N}\sigma _{j}^{z}$ was
defined which differs from $H_{N}=q|0\rangle \langle 0|$ which simply is the
projection onto the unique basis vector in $V_{0,q}$. In \emph{loc. cit.}
the decomposition $V_{q}^{\otimes N}=\bigoplus_{n=0}^{N}V_{n,q}$ has been
employed to define via the eigenbasis $\{\mathfrak{e}_{\lambda }\}_{\lambda
\in (n,k)}$ the following set of operators 
\begin{equation*}
\hat{H}_{r}|V_{n,q}=\sum_{\lambda \in (n,k)}h_{r}(y(\lambda ))\frac{%
\mathfrak{e}_{\lambda }\otimes \mathfrak{e}_{\lambda }^{\ast }}{||\mathfrak{e%
}_{\lambda }||^{2}},\;r\geq 0\;.
\end{equation*}%
Clearly, we have $H_{r}=\hat{H}_{r}$ for $0\leq r<N$. The new insight here
is that these operators originate from a Yang-Baxter algebra which has a
local description in terms of a statistical vertex model on a square lattice.
\end{remark}

\subsection{The Verlinde algebra}

Since the $q$-dependence can be removed via a simple rescaling of the Bethe
roots, $y\rightarrow q^{-\frac{1}{N}}y$ (compare with (\ref{bae})) we now
set for simplicity $q=1$. Interpret the eigenbasis $\{\mathfrak{e}_{\lambda
}\}_{\lambda \in (n,k)}$ as a complete set of orthogonal idempotents of an
associative, unital and commutative algebra. Next we show that the resulting
generalised matrix algebra is isomorphic to the Verlinde algebra $\mathcal{F}%
_{n,k}^{\mathbb{C}}$, where we recall from the introduction that $\mathcal{F}%
_{n,k}\cong qH^{\ast }(\func{Gr}\nolimits_{n,n+k})/\langle q-1\rangle $. 

\begin{theorem}[idempotents of the Verlinde algebra]
Endow $V_{n}$ with the following product, $\mathfrak{e}_{\lambda }\star 
\mathfrak{e}_{\mu }:=\delta _{\lambda \mu }||\mathfrak{e}_{\lambda }||^{2}%
\mathfrak{e}_{\lambda }$ and bilinear form $\eta (|\lambda \rangle ,|\mu
\rangle )=\delta _{\lambda ^{\vee }\mu }$. Then

\begin{enumerate}
\item $(V_{n},\star ,\eta)$ is a commutative Frobenius algebra.

\item The map $|\lambda \rangle \mapsto s_{\lambda }$ is an algebra
isomorphism $(V_{n},\star )\cong \mathcal{F}_{n,k}^{\mathbb{C}}$.
\end{enumerate}
\end{theorem}

\begin{remark}
An analogous statement holds true for the $\mathfrak{\widehat{su}}(n)_k$%
-WZNW fusion ring using so-called $\infty$-friendly walkers; see \cite{KS}.
In \cite[Section 5]{Korffproc} the role of the Bethe vectors as idempotents
has been highlighted and Section 7 of \emph{loc. cit.} explains how the
construction might generalise to other integrable models.
\end{remark}

\begin{proof}
We compute the product expansion in the basis $\{|\lambda \rangle \}$.
Exploiting (\ref{Bethevector}) we find that 
\begin{eqnarray}
|\lambda \rangle \star |\mu \rangle &=&\sum_{\alpha ,\beta \in (n,k)}\langle 
\mathfrak{e}_{\alpha }^{\ast }|\lambda \rangle \langle \mathfrak{e}_{\beta
}^{\ast }|\mu \rangle \mathfrak{e}_{\alpha }\star \mathfrak{e}_{\beta } 
\notag \\
&=&\sum_{\alpha \in (n,k)}\frac{s_{\lambda }(y(\alpha ))s_{\mu }(y(\alpha ))%
}{||\mathfrak{e}_{\alpha }||^{2}}\mathfrak{e}_{\alpha }=\sum_{\nu \in
(n,k)}C_{\lambda \mu }^{\nu ,d}|\nu \rangle ,  \label{Schurprod1}
\end{eqnarray}%
%
%
%
%
%
%
%
%
%
%
%
%
%
%
where in the last step we have used, once more, the definition (\ref%
{Bethevector}) and the Bertram-Vafa-Intriligator formula for Gromov-Witten
invariants (see e.g. \cite[Cor 6.2, Eqn (6.1)]{Rietsch}), 
\begin{equation}
C_{\lambda \mu }^{\nu ,d}=\sum_{\alpha \in (n,k)}\frac{s_{\lambda }(y(\alpha
))s_{\mu }(y(\alpha ))s_{\nu }(y(\alpha )^{-1})}{||\mathfrak{e}_{\alpha
}||^{2}}\;.  \label{Schurprod2}
\end{equation}%
From this equality it is obvious that $|\emptyset \rangle $ is the identity
with respect to the product $\star $ and using (\ref{Bethevector}) we find 
\begin{equation}
|\emptyset \rangle =\sum_{\alpha \in (n,k)}\mathfrak{\hat{e}}_{\alpha
},\qquad \mathfrak{\hat{e}}_{\alpha }=\mathfrak{e}_{\alpha }/||\mathfrak{e}%
_{\alpha }||^{2}\;.
\end{equation}%
We now turn to the Frobenius structure. That $\eta $ is non-degenerate
follows from the observation that $|\lambda \rangle \mapsto |\lambda ^{\vee
}\rangle $ simply permutes the basis elements in $V_{n}$. Compatibility of $%
\eta $ with the product amounts to the identity 
\begin{equation}
\eta (|\lambda \rangle \star |\mu \rangle ,|\nu \rangle )=C_{\lambda \mu \nu
}=C_{\mu \nu \lambda }=\eta (|\lambda \rangle ,|\mu \rangle \star |\nu
\rangle ),
\end{equation}%
where $C_{\lambda \mu \nu }:=C_{\lambda \mu }^{\nu ^{\vee },d}$ and we have
used the known $S_{3}$-invariance of Gromov-Witten invariants, $C_{\lambda
\mu \nu }=C_{\pi (\lambda )\pi (\mu )\pi (\nu )}$ for all $\pi \in S_{3}$,
which is immediate from their geometric definition; see e.g. \cite{BCF}.
\end{proof}

Our main motivation to emphasise the Frobenius structure is the connection
with the toric Schur polynomials mentioned after Prop \ref{Frob2toric} in
the introduction.

\begin{proof}[\textbf{Proof of Prop \protect\ref{Frob2toric}}]
Let $\mathfrak{m}:\mathcal{F}_{n,k}\otimes \mathcal{F}_{n,k}\rightarrow 
\mathcal{F}_{n,k}$ be the regular representation or multiplication map, $%
\mathfrak{m}( s_\mu \otimes s_\nu )= s_\mu \star s_\nu$, and $\mathfrak{m}%
^{\ast }:\mathcal{F}_{n,k}^{\ast }\rightarrow \mathcal{F}_{n,k}^{\ast
}\otimes \mathcal{F}_{n,k}^{\ast }$ its dual map with the Frobenius
isomorphism $\Phi :\mathcal{F}_{n,k}\rightarrow \mathcal{F}_{n,k}^{\ast }$
given by $\Phi : s_\lambda \mapsto \eta( s_{\lambda },\bullet)$. Then we
recall that the co-product $\Delta _{n,k}$ is obtained via the following
commutative diagram, 
\begin{equation}
\begin{CD} \mathcal{F}_{n,k} @>\Delta_{n,k}>>
\mathcal{F}_{n,k}\otimes\mathcal{F}_{n,k}\\ @VV{\Phi}V
@VV{\Phi\otimes\Phi}V\\ \mathcal{F}^\ast_{n,k} @>\mathfrak{m}^{\ast}>>
\mathcal{F}^\ast_{n,k}\otimes\mathcal{F}^\ast_{n,k} \end{CD}\quad .
\label{diagram}
\end{equation}%
Thus, we compute $\mathfrak{m}^\ast\circ\Phi( s_\lambda)( s_\mu\otimes
s_\nu)= \Phi( s_\lambda)( s_\mu\star s_\nu)=C_{\lambda\mu\nu}$. Assuming on
the other hand that $\Delta_{n,k} s_\lambda=\sum_{d,\alpha}
s_{\lambda/d/\alpha}\otimes s_\alpha$ with $s_{\lambda/d/\alpha}=\sum_{%
\beta}C_{\alpha\beta}^{\lambda,d} s_\beta$ we find 
\begin{equation*}
(\Phi\otimes\Phi)\Delta_{n,k}( s_\lambda)( s_\mu\otimes s_\nu)=
\sum_{d,\alpha,\beta}C_{\alpha\beta}^{\lambda,d}(\Phi( s_\beta)\otimes\Phi(
s_\alpha))( s_\mu\otimes s_\nu)= C_{\lambda^\vee\mu^\vee\nu^\vee}\;.
\end{equation*}
Exploiting invariance under Poincar\'e duality (see e.g. \cite{BCF}), $%
C_{\lambda^\vee\mu^\vee\nu^\vee}=C_{\lambda\mu\nu}$, the assertion now
follows.
\end{proof}

%

%
%
%

\section{Bijections between walks and toric tableaux\label{sec:bijections}}

In this section we prove that the lattice configurations of the vicious and
osculating walker models are in bijection with toric tableaux which were
used in \cite{BCF} and \cite{Postnikov}. This will provide a combinatorial
link between these statistical mechanics models and the quantum cohomology
ring: the partition functions of vicious and osculating walkers on the
cylinder are toric Schur functions.

To keep this article self-contained we start by recalling the definition of
toric tableaux following \cite{Postnikov}. As mentioned earlier, toric
tableaux are a particular subset of cylindric tableaux; see \cite%
{GesselKrattenthaler} for the original definition of cylindric (plane)
partitions, lattice paths and cylindric Schur functions.

\begin{definition}[cylindric loops]
Let $\lambda =(\lambda _{1},\ldots ,\lambda _{n})\in (n,k)$ and define the
following associated \emph{cylindric loops} $\lambda \lbrack r]$ for any $%
r\in \mathbb{Z}$,%
\begin{equation}
\lambda \lbrack r]:=(\ldots ,\underset{r}{\lambda _{n}+r+k},\underset{r+1}{%
\lambda _{1}+r},\ldots ,\underset{r+n}{\lambda _{n}+r},\underset{r+n+1}{%
\lambda _{1}+r-k},\ldots )\;.  \notag
\end{equation}
\end{definition}

For $r=0$ the cylindric loop can be visualized as a path in $\mathbb{Z}%
\times \mathbb{Z}$ determined by the outline of the Young diagram of $%
\lambda $ which is periodically continued with respect to the vector $(n,-k)$%
. For $r\neq 0$ this line is shifted $r$ times in the direction of the
lattice vector $(1,1)$; see Figure \ref{fig:cylindricloop} for an
illustration.

\begin{figure}[tbp]
\begin{equation*}
\includegraphics[scale=0.45]{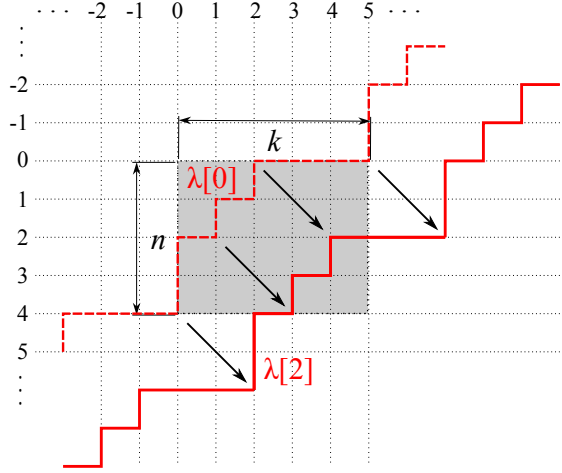}
\end{equation*}%
\caption{Example of a cylindric loop.}
\label{fig:cylindricloop}
\end{figure}

\begin{definition}[cylindric skew diagrams]
Given two partitions $\lambda ,\mu \in (n,k)$ denote by $\lambda /d/\mu $
the set of squares between the two lines $\lambda \lbrack d]$ and $\mu
\lbrack 0]$ modulo integer shifts by $(n,-k)$, 
\begin{equation*}
\lambda /d/\mu :=\{\langle i,j\rangle \in \mathbb{Z}\times \mathbb{Z}/(n,-k)%
\mathbb{Z}~|~\lambda \lbrack d]_{i}\geq j>\mu \lbrack 0]_{i}\}\;.
\end{equation*}%
We shall refer to $\lambda /d/\mu $ as a \emph{cylindric skew-diagram} of
degree $d$.
\end{definition}

A cylindric skew diagram $\nu /d/\mu $ which has at most one box in each
column will be called a (cylindric) \emph{horizontal strip} and one which
has at most one box in each row a (cylindric) \emph{vertical strip}. The 
\emph{length} of such strips will be the number of boxes within the skew
diagram.

\begin{definition}[toric skew diagrams]
A cylindric skew diagram is called toric if it has at most $k$ squares in
each row.
\end{definition}

See Figure \ref{fig:toric_tableau} for an example of a toric skew-diagram.
Note that for $d=0$ we recover the familiar skew-diagram of two partitions,
i.e. $\lambda /0/\mu =\lambda /\mu $.

\begin{definition}[cylindric tableaux]
A (semi-standard) tableau $T$ of cylindric shape $\lambda /d/\mu $ is a
mapping $T:\lambda /d/\mu \mapsto \mathbb{N}$ of the squares of the
associated diagram such that in each row and column of connected squares the
numbers are respectively weakly increasing (left to right) and strictly
increasing (top to bottom).
\end{definition}

See once more Figure \ref{fig:toric_tableau} for an example. As for ordinary
skew-diagrams we define the weight vector $\alpha (T)=(\alpha _{1},\alpha
_{2},\ldots )$ by setting $\alpha _{i}$ to be the number of $i$-entries in $%
T $.

\begin{figure}[tbp]
\begin{equation*}
\includegraphics[scale=0.45]{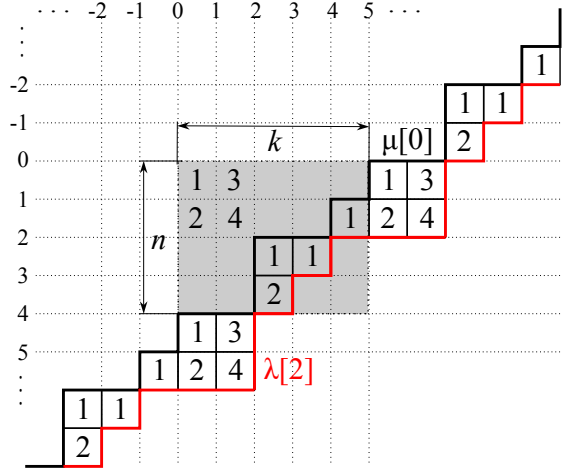}
\end{equation*}%
\caption{Example of a toric tableau. On the left it is indicated how the
toric skew diagram is constructed from the cylindric loops of $\protect%
\lambda $ and $\protect\mu $. The right picture shows how the squares
outside the bounding box are moved to the inside.}
\label{fig:toric_tableau}
\end{figure}

\begin{definition}[quantum Kostka numbers]
The cardinality of the set of all cylindric tableaux $T$ of shape $\nu/d/\mu$
and weight $\alpha$ is denoted by $K_{\nu/d/\mu,\alpha}$.
\end{definition}

As already pointed out in the introduction the \emph{quantum Kostka number} $%
K_{\nu /d/\mu ,\alpha }$ in (\ref{qKostkadef}) equals the number of
semi-standard cylindric tableaux of weight $\alpha $, and specialises for $%
d=0$ to the ordinary Kostka number $K_{\nu /\mu ,\alpha }$.

Below we will also make use of $K_{\nu ^{\prime }/d/\mu ^{\prime },\alpha }$%
, the number of conjugate cylindric and toric diagrams $\nu ^{\prime }/d/\mu
^{\prime }$, which are obtained by interchanging $n$ and $k$. Here $\nu
^{\prime },\mu ^{\prime }\in (k,n)$ denote the conjugate partitions of $\nu
,\mu $.

\begin{definition}[cylindric Schur functions]
Introduce the following generalisation of a skew Schur function,%
\begin{equation}
s_{\nu /d/\mu }(x_{1},x_{2},\ldots ):=\sum_{|T|=\nu /d/\mu
}x^{T}=\sum_{\alpha }K_{\nu /d/\mu ,\alpha }m_{\alpha }(x_{1},x_{2},\ldots
)\,,  \label{cylindricschur}
\end{equation}%
where $m_{\alpha }$ are the monomial symmetric functions in an infinite set
of variables $x_{i}$ and the sum runs over all cylindric tableaux of fixed
shape $\nu/d/\mu$.
\end{definition}

Note that the cylindric Schur function (\ref{cylindricschur}) specialises to
an ordinary skew Schur function for $d=0$.

\begin{definition}[toric Schur functions]
Specialising $x_{r}=0,~r>n$ the cylindric Schur functions are called \emph{%
toric} Schur functions, i.e. the sum in (\ref{cylindricschur}) runs only
over toric tableaux.
\end{definition}

We recall from \cite[Lemma 5.2]{Postnikov} that toric Schur functions are
nonzero if and only if $\nu /d/\mu $ is a toric skew diagram.

\subsection{Lattice configurations and quantum Kostka numbers}

Throughout this section we assume $\nu,\mu\in(n,k)$. Denote by $\Gamma _{\nu
,\mu },\Gamma _{\nu ,\mu }^{\prime }$ the sets of all allowed lattice
configurations $\mathcal{C},\,\mathcal{C}^{\prime }$ on the cylinder for the
vicious and osculating walker models where the values of the lower and upper
vertical lattice edges are fixed by the 01-words $w(\nu )$ and $w(\mu )$,
respectively. It will also be convenient to consider the subsets $\Gamma
_{\nu /d/\mu }\subset \Gamma _{\nu ,\mu }$ of configurations which do have a
fixed number of $2d$ outer horizontal edges with value one (they come in
pairs due to the quasi-periodic boundary conditions in the horizontal
direction), including in particular the special case of $d=0$ ($q=0$) when
there are no outer horizontal edges $\Gamma _{\nu /0/\mu }$. Finally, we
introduce for each $\alpha =(\alpha _{1},\ldots ,\alpha _{n})\in \mathbb{Z}%
_{\geq 0}^{n}$ the subsets 
\begin{equation}
\Gamma _{\nu /d/\mu }(\alpha ):=\{\mathcal{C}\in \Gamma _{\nu /d/\mu }:\func{%
wt}(\mathcal{C})=x^{\alpha }\}\,.
\end{equation}%
Analogously we define for $\beta =(\beta _{1},\ldots ,\beta _{k})\in \mathbb{%
Z}_{\geq 0}^{k}$ the set $\Gamma _{\nu /d/\mu }^{\prime }(\beta )$ as the
lattice configurations $\mathcal{C}^{\prime }$ of the osculating walker
model which have weight $\func{wt}^{\prime }(\mathcal{C}^{\prime })=x^{\beta
}$.

\begin{proposition}
\label{prop:bijection} The set of allowed lattice configurations $\Gamma
_{\nu /d/\mu }$ and $\Gamma _{\nu /d/\mu }^{\prime }$ are in bijection with
the sets of toric skew tableaux of shape $\nu /d/\mu $ and $\nu ^{\prime
}/d/\mu ^{\prime }$, respectively. In particular, 
\begin{equation}
K_{\nu /d/\mu ,\alpha }=|\Gamma _{\nu /d/\mu }(\alpha )|=|\Gamma _{\nu
^{\prime }/d/\mu ^{\prime }}^{\prime }(\alpha )|\,,
\end{equation}%
where $\alpha =(\alpha _{1},\ldots ,\alpha _{n})$ is some weight vector with
non-negative integer entries.
\end{proposition}

Note that for $d=0$ and $\mu=\emptyset$ the statement specialises to Claim %
\ref{claim} which we used earlier to prove Lemma \ref{lem:Bethe}. 
\begin{figure}[tbp]
\begin{equation*}
\includegraphics[scale=0.29]{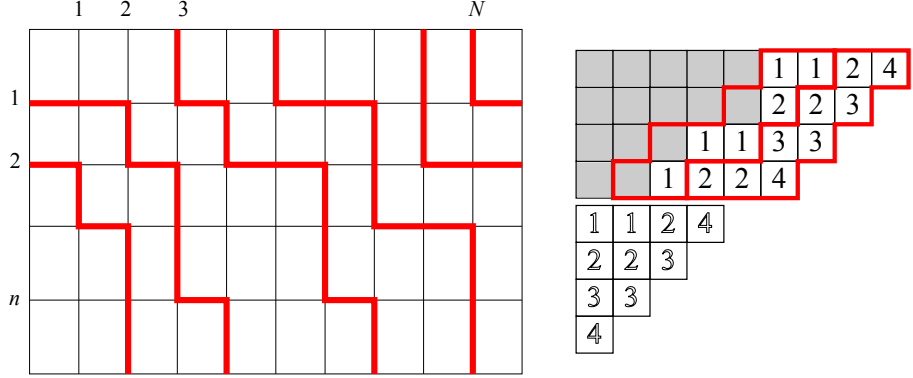}
\end{equation*}%
\caption{Example for constructing a toric tableau from the vicious walker
configuration considered in the introduction with $N=9=n+k=4+5$ and $\protect%
\lambda=(5,5,3,2)=\protect\mu,\,\protect\nu=(5,4,2,1)$. The shaded boxes on
the right show the Young diagram of $\protect\mu$ and removing from the
toric skew diagram the two ($d=2$) framed $N$-ribbons, one obtains the Young
diagram of $\protect\nu$. The boxes below belong to the periodic
continuation of the tableau.}
\label{fig:walk_tableau_bijection}
\end{figure}

\begin{proof}
We concentrate on the vicious walker model, the generalisation to the
osculating walker model will then be obvious. First we state the bijection.
Recall from Figure \ref{fig:QCvertex} that each lattice configuration $%
\mathcal{C}\in \Gamma _{\nu /d/\mu }$ defines an $n$-tuple of
non-intersecting paths $\gamma =(\gamma _{1},\ldots ,\gamma _{n})$ of which $%
d$ cross the boundary. Draw the Young diagram $Y(\mu )$ of $\mu $ in the
bounding box $(n,k)$. Reading the $01$-word of $\mu $ from left to right
take the path which originates at $\ell _{1}(\mu )$, that is from the first
one letter in $w(\mu )$, and note down the lattice rows of each of its
horizontal edges starting from the top, say $i_{1}\leq \cdots \leq i_{r}$.
Add a corresponding row of boxes with entries $i_{1}\leq \cdots \leq i_{r}$
to the bottom row of $Y(\mu )$. If the path has no horizontal edges do not
add any boxes. Continue with the path originating from the second 1-letter
in $w(\mu )$ at $\ell _{2}(\mu )$ and write the corresponding filled boxes
in the row above the bottom row of $Y(\mu )$ starting at the first square
which does not lie in $\mu $. Continue until you have reached the last
1-letter in $w(\mu )$; see Figure \ref{fig:walk_tableau_bijection} for an
example.

That the described map is indeed a bijection follows from the following
lemma.

\begin{lemma}
Each toric skew tableau $T$ of shape $\nu /d/\mu $ can be written as a
sequence of toric horizontal strips, i.e. there is a unique sequence of
cylindric loops $(\mu \lbrack 0]=\nu ^{(0)}[d_{0}],\nu ^{(1)}[d_{1}],\ldots
,\nu ^{(n)}[d_{n}]=\nu \lbrack d])$ such that $\nu
^{(i+1)}/(d_{i+1}-d_{i})/\nu ^{(i)}$ is a horizontal strip and $%
\sum_{i}d_{i}=d$ with $d_{i}=0,1$.
\end{lemma}

Since toric tableaux can be seen as a special subset of ordinary
(semi-standard) tableaux of shape $\nu[d]/\mu$ the proof of this lemma
follows along very similar lines as in the case of ordinary tableaux (see
e.g. \cite[Chap I]{Macdonald}) and we therefore omit it.

Thus, it suffices to derive the assertion for $n=1$ in which case $d=0$ or $%
1 $. Then we have the following generalisation of (\ref{part2word}) to
cylindric loops%
\begin{equation*}
\mu \lbrack 0]_{n+1-i}=\ell _{i}(\mu )-i\text{\qquad and\qquad }\lambda
\lbrack 1]_{n+1-i}=\ell _{i+1}(\lambda )-i\;,
\end{equation*}%
where $\ell _{i+n}(\mu )=\ell _{i}(\mu )+N$, $\ell _{i+n}(\lambda )=\ell
_{i}(\lambda )+N$. Analogously, one finds for the conjugate loops, 
\begin{equation*}
\mu \lbrack 0]_{k+1-i}^{\prime }=\mu ^{\prime }[0]_{k+1-i}=\ell _{i}(\mu
^{\prime })-i,\text{\qquad }\lambda \lbrack 1]_{k+1-i}^{\prime }=\lambda
^{\prime }[1]_{k+1-i}=\ell _{i+1}(\lambda ^{\prime })-i
\end{equation*}%
with $\ell _{i+k}(\mu ^{\prime })=\ell _{i}(\mu ^{\prime })+N$, $\ell
_{i+k}(\lambda ^{\prime })=\ell _{i}(\lambda ^{\prime })+N$. Using these
formulae together with the action (\ref{TLrep}), (\ref{affTLrep}), one now
easily verifies that an allowed row configuration of the vicious walker
model, i.e. a non-vanishing matrix element of a monomial in the hopping
operators $f_i$ appearing in (\ref{H}), defines a toric horizontal strip $%
\nu /d/\mu $ with $d=0,1$ and vice versa. In particular, the action of a
consecutive string such as $f_{i}\cdots f_{2}f_{1}f_{N}\cdots f_{j}b_{\mu
}=b_{\lambda }$ can only be nonzero if there are as many consecutive
0-letters in $w(\mu )$ starting at $j$ and ending at $i$ as there are
hopping operators $f_{l}$ in the string. Thus, the horizontal strip has at
most length $k$.

The bijection for the osculating walkers is analogous. Start with the
leftmost path originating at $\ell _{1}(\mu )$ and for each horizontal path
edge in lattice row $i$ add a box labelled $i$ in the rightmost column of
the $k\times n$ bounding box beneath the Young diagram of $\mu ^{\prime }$.
Continue with the second path placing the boxes now in the second column
from the right and so forth. The result is a conjugate toric tableau of
shape $\nu ^{\prime }/d/\mu ^{\prime }$; see Figure \ref%
{fig:walk_tableau_bijection_2} for an example.

The proof that the described map is indeed a bijection employs the expansion
(\ref{E}) and follows closely along similar lines as in the previous case of
vicious walkers. We therefore omit it.
\end{proof}


\begin{figure}[tbp]
\begin{equation*}
\includegraphics[scale=0.29]{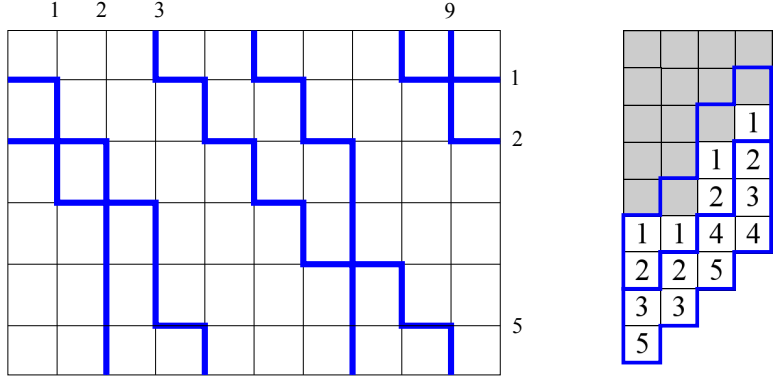}
\end{equation*}%
\caption{Example for constructing a conjugate toric tableau from the
osculating walker configuration considered in the introduction. The shaded
boxes on the right now show the Young diagram of $\protect\mu^{\prime}$ and
removing from the toric skew diagram the two $N$-ribbons ($d=2$) framed in
blue, one obtains the Young diagram $\protect\nu^{\prime}$. }
\label{fig:walk_tableau_bijection_2}
\end{figure}


Number the lattice columns of $\mathbb{L}$ ( $\mathbb{L}^{\prime }$) from
left to right. That is, the $i$th column is the collection of \emph{%
horizontal} lattice edges $(p,p^{\prime })\in \mathbb{E}$ ( $\mathbb{E}%
^{\prime }$) such that $p_{1}=i$ and $p_{1}^{\prime }=i+1$. Let $\mathcal{C}:%
\mathbb{E}\rightarrow \{0,1\}$ ($\mathcal{C}^{\prime }:\mathbb{E}^{\prime
}\rightarrow \{0,1\}$) be a vicious (osculating) walker configuration with
start and end positions $\ell (\nu )$ and $\ell (\mu )$, $\nu ,\mu \in (n,k)$%
. The number of horizontal \emph{path} edges in column $i$ is the sum over
the \emph{values} of the horizontal lattice edges in column $i$ in
configuration $\mathcal{C}$ ($\mathcal{C}^{\prime }$).

\begin{lemma}
The number of horizontal path edges in lattice column $i$ is 
\begin{equation}
\theta _{i}(\nu ,\mu ,d)=d+n_{i}(\mu )-n_{i}(\nu ),\quad \quad 1\leq i\leq
N\;.  \label{theta}
\end{equation}%
In particular, we have that $d\geq d_{\min }(\nu ,\mu ):=\max_{i\in
I}\{n_{i}(\nu )-n_{i}(\mu )\}$.
\end{lemma}

\begin{remark}
The integer $d_{\min }(\nu ^{\vee },\mu )$ is the minimal power appearing in
the quantum product $s_{\mu }\ast s_{\nu }$; see \cite{FW}. In fact, there
exists an interval of integers $[d_{\min },d_{\max }]$ describing all powers
occurring in this product. This was first conjectured in \cite{Yong} and
proved in \cite{Postnikov}.
\end{remark}

\begin{proof}
Denote by $w_{i}^{(j)},v_{i}^{(j)}\in \{0,1)$ the values of the respectively
horizontal and vertical lattice edges in column $i$ and row $j$ in the
configuration $\mathcal{C}$ or $\mathcal{C}^{\prime }$. Since the argument
is completely analogous for both, vicious and osculating walkers, we only
consider the former for the rest of the proof. By construction we have $%
v_{i}^{(0)}=w_{i}(\mu )$ and $v_{i}^{(n)}=w_{i}(\nu )$. Note that the
allowed configurations shown in Figure \ref{fig:QCvertex} preserve the sum
of values on the N and W edge and the E and S edge, that is $%
w_{i-1}^{(j)}+v_{i}^{(j-1)}=v_{i}^{(j)}+w_{i}^{(j)}$. Hence, we compute%
\begin{eqnarray*}
\theta _{i} &=&w_{i}^{(1)}+w_{i}^{(2)}+\cdots +w_{i}^{(n)} \\
&=&v_{i}^{(0)}+w_{i-1}^{(1)}-v_{i}^{(1)}+v_{i}^{(1)}+w_{i-1}^{(2)}-v_{i}^{(2)}+\cdots +v_{i}^{(n-1)}+w_{i-1}^{(n)}-v_{i}^{(n)}
\\
&=&\theta _{i-1}+v_{i}^{(0)}-v_{i}^{(n)}=\theta _{i-1}+w_{i}(\mu )-w_{i}(\nu
)\;.
\end{eqnarray*}%
The assertion now follows by observing that $\theta _{N}=d$, the number of
horizontal path edges on the boundary.
\end{proof}

%
%
%

\begin{lemma}
\label{degree} Let $\lambda\in(n,k)$. $K_{\nu /d/\mu ,\lambda }=0$ unless $%
|\lambda |+|\mu |-|\nu |=dN$.
\end{lemma}

\begin{proof}
Since $\theta _{N}(\nu ,\mu ,d)=d$ and the vicious and osculating paths are
non-intersecting we must have that the paths starting at $\ell _{N+1-d}(\mu
),\ldots ,\ell _{N}(\mu )$ end at $\ell _{1}(\nu ),\ldots ,\ell _{d}(\nu )$,
respectively. Thus, the paths starting at $\ell _{1}(\mu ),\ldots ,\ell
_{N-d}(\mu )$ end at $\ell _{d+1}(\nu ),\ldots ,\ell _{N}(\nu )$. Therefore, 
%
$|\lambda |=\sum_{i=1}^{n}(\ell _{i+d}(\nu )-\ell _{i}(\mu ))=|\nu |-|\mu
|+dN$. 
\end{proof}

\subsection{Path counting and Gromov-Witten invariants}

We now compute the weighted sums over lattice configurations of the vicious
and osculating walker models; these are called \emph{partition functions}.
Consider the following operator products $H(x_{1})\cdots H(x_{n})\in \mathbb{%
Z}[x_{1},\ldots ,x_{n}]\otimes \func{End}(V_{q}^{\otimes N})$ and $%
E(x_{1})\cdots E(x_{k})\in \mathbb{Z}[x_{1},\ldots ,x_{k}]\otimes \func{End}%
(V_{q}^{\otimes N})$. Then by construction we have for the vicious walker
model%
\begin{equation}
\langle \nu |H(x_{1})\cdots H(x_{n})|\mu \rangle =\sum_{\mathcal{C}\in
\Gamma _{\nu ,\mu }}\prod_{\langle i,j\rangle \in \mathbb{L}}\func{wt}(%
\mathrm{v}_{i,j})\in \mathbb{Z}[x_{1},\ldots ,x_{n}],  \label{5vZ}
\end{equation}%
and for the osculating walker model, 
\begin{equation}
\langle \nu |E(x_{1})\cdots E(x_{k})|\mu \rangle =\sum_{\mathcal{C}^{\prime
}\in \Gamma _{\nu ,\mu }^{\prime }}\prod_{\langle i,j\rangle \in \mathbb{L}%
^{\prime }}\func{wt}\nolimits^{\prime }(\mathrm{v}_{i,j})\in \mathbb{Z}%
[x_{1},\ldots ,x_{k}],  \label{5vZ'}
\end{equation}%
where $\langle \nu |X|\mu \rangle $ is shorthand for $\langle b(\nu ),Xb(\mu
)\rangle $ and $\func{wt}(\mathrm{v}_{i,j})$, $\func{wt}^{\prime }(\mathrm{v}%
_{i,j})$ denote the vertex weights of the vicious and osculating walker
models in Figures \ref{fig:QCvertex}, \ref{fig:QCvertex2}.

\begin{corollary}
The partition function (\ref{5vZ}) has the following expansion in toric
Schur polynomials 
\begin{equation}
\langle \nu |H(x_{1})\cdots H(x_{n})|\mu \rangle =\sum_{d\geq
0}q^{d}s_{\nu/d/\mu }(x_{1},\ldots ,x_{n})\,.  \label{ZgeneratesGW}
\end{equation}%
The analogous expansion of the partition function (\ref{5vZ'}) for the
osculating walker model reads%
\begin{equation}
\langle \nu |E(x_{1})\cdots E(x_{k})|\mu \rangle =\sum_{d\geq
0}q^{d}s_{\nu^{\prime}/d/\mu^{\prime}}(x_{1},\ldots ,x_{k})\;.
\label{Z'generatesGW}
\end{equation}%
In particular, for $\lambda\in(n,k)$ the quantum Kostka numbers are given by
the matrix elements 
\begin{equation}
q^dK_{\nu /d/\mu ,\lambda }=\langle \nu |H_{\lambda }|\mu \rangle =\langle
\nu ^{\prime }|E_{\lambda }|\mu ^{\prime }\rangle \;,  \label{Kmatrix}
\end{equation}
where $H_{\lambda }=H_{\lambda _{1}}\cdots H_{\lambda _{n}}$, $E_{\lambda
}=E_{\lambda _{1}}\cdots E_{\lambda _{n}}$ and $d=(|\lambda|+|\mu|-|\nu|)/N$.
\end{corollary}

\begin{proof}
Let $m_\lambda$ be the monomial symmetric function. We have the identities 
\begin{eqnarray}
\langle \nu |H(x_{1})\cdots H(x_{n})|\mu \rangle &=&
\sum_{\lambda\in(n,k)}\langle \nu |H_{\lambda }|\mu \rangle m_{\lambda
}(x_{1},\ldots ,x_{n})  \label{Cauchy1} \\
&=& \sum_{\lambda\in(n,k)}q^{d(\lambda)}K_{\nu /d/\mu ,\lambda }m_{\lambda
}(x_{1},\ldots ,x_{n}) ,  \label{Cauchy1a}
\end{eqnarray}%
where (\ref{Cauchy1}) is a direct consequence of the expansion (\ref{H}) and
the fact that the partition function (\ref{5vZ}) must be symmetric in the $%
x_{i}$'s due to $H(x)H(y)=H(y)H(x)$. The asserted identity (\ref{Cauchy1a})
with $d(\lambda)N=|\lambda|+|\mu|-|\nu|$ follows from the bijection
described in the proof of Prop \ref{prop:bijection} and Lemma \ref{degree}.
Recalling that the monomial symmetric functions form a basis in the ring of
symmetric function the last equality is proved. The argument is completely
analogous for the osculating walkers.
\end{proof}

We now recover the following result from \cite[Thm 5.3]{Postnikov} and \cite[%
Thm 10.8]{KS}.

\begin{corollary}
\label{cor:toric} Let $n\neq 0,N$. One has the following expansion of toric
Schur functions,%
\begin{equation}
s_{\nu /d/\mu }(x_{1},\ldots ,x_{n})=\sum_{|\lambda |+|\mu |-|\nu
|=dN}C_{\lambda \mu }^{\nu ,d}s_{\lambda }(x_{1},\ldots ,x_{n})\,,
\label{toricschur}
\end{equation}%
where $C_{\lambda \mu }^{\nu ,d}=\langle \nu |S_{\lambda }|\mu \rangle $ and 
$S_{\lambda }=\det (H_{\lambda _{i}-i+j})_{1\leq i,j\leq n}$.
\end{corollary}

\begin{proof}
Recall from the ring of symmetric functions the familiar expansion $%
m_{\alpha }=\sum_{\lambda }L_{\alpha \lambda }s_{\lambda }$, where $%
L_{\alpha \lambda }=\sum_{w}^{\prime }(-1)^{w}$ with the sum ranging over
all permutations $w\in S_{n}$ such that $w(\lambda +\rho )-\rho \in
S_{n}\alpha $. Here $\rho =(n-1,\ldots ,2,1,0)$ is the Weyl vector.
Employing this identity together with $S_{\lambda
}=\sum_{w}(-1)^{w}H_{w(\lambda +\rho )-\rho }$, one proves that the
expansion coefficients of the partition function (\ref{ZgeneratesGW}) into
Schur functions are the matrix elements $\langle \nu |S_{\lambda }|\mu
\rangle $. That the latter equal the Gromov-Witten invariants can be derived
from Cor \ref{cor:iso} and is a consequence of (\ref{Schurprod1}), (\ref%
{Schurprod2}) together with the eigenvalue equations (\ref{Hspec}), (\ref%
{Espec}) and Theorem \ref{completeness}.
\end{proof}

Let $c(s)=j-i$ be the \emph{content} and $h(s)=\lambda _{i}+\lambda
_{j}^{\prime }-i-j+1$ the \emph{hook-length} of a square $s=(i,j)\in \lambda 
$. As a special case of the last corollary we obtain the following solution
to the counting problem of non-intersecting paths on the cylinder which is a
refinement of the one stated in the introduction.

\begin{corollary}
Setting $x_{1}=\cdots =x_{n}=1$ we have that 
\begin{equation}
|\Gamma _{\nu /d/\mu }|=\sum_{\substack{ \lambda \in (n,k)  \\ |\lambda
|+|\mu |-|\nu |=dN}}C_{\lambda \mu }^{\nu ,d}\prod_{s\in \lambda }\frac{%
n+c(s)}{h(s)}=|\Gamma _{\nu ^{\prime }/d/\mu ^{\prime }}^{\prime }|\;.
\label{Zcardinality}
\end{equation}%
%
%
%
%
%
%
%
%
%
%
%
%
%
%
%
%
%
%
%
%
%
%
%
%
%
%
%
%
%
%
\end{corollary}

\begin{proof}
Trivial consequence of the known identity $s_{\lambda }(1,\ldots
,1)=\prod_{s\in \lambda }\frac{n+c(s)}{h(s)}$ for Schur functions, see e.g. 
\cite[Chapter I, Ex 1(a), page 26]{Macdonald}. The above result then simply
states that setting $q=1$ we have%
\begin{equation}
|\Gamma _{\nu ,\mu }|=\langle \nu |H^{n}|\mu \rangle =\sum_{d\geq 0}s_{\nu
/d/\mu }(1,\ldots ,1)=\langle \nu ^{\prime }|E^{k}|\mu ^{\prime }\rangle
=|\Gamma _{\nu ^{\prime },\mu ^{\prime }}^{\prime }|\;
\end{equation}%
where $H=H(1)$ and$\ E=E(1)$.
\end{proof}

\begin{remark}
For $q=0$ we recover the result for configurations on the finite open strip
by applying the Gessel-Viennot Theorem; see e.g. \cite{Forrester,Forrester2}%
, \cite{BrakOw} and \cite{Guttmann1,Guttmann2,Guttmann3}. This yields the
familiar determinant (Jacobi-Trudi) formula for skew Schur functions. 
For general $q$ the only known expressions involve sums of determinants; see 
\cite[Equation (6.4), page 299]{McNamara} and \cite[Prop 1]%
{GesselKrattenthaler}.
\end{remark}

\section{Toric Schur function identities}

\label{Sec:Demazure}

We are now in the position to derive the formula (\ref{toricDemazure}) and
the generating function (\ref{toric_generating}) stated in the introduction.
We start with proving the identity (\ref{cop2}).

\begin{proof}[Proof of Proposition \protect\ref{Prop:cop}]
It follows from our previous results (\ref{ZgeneratesGW}) and (\ref{Bethev1}%
) that we have the following expansion,%
\begin{equation*}
\langle \lambda |H(y_{n})\cdots H(y_{1})B(x_{n})\cdots B(x_{1})|0\rangle
=x^{\delta _{n}}\sum_{\substack{ \mu \in (n,k)  \\ d\geq 0}}q^{d}s_{\lambda
/d/\mu }(y)s_{\mu }(x)\;.
\end{equation*}%
But according to Prop \ref{Frob2toric} the right hand side - apart from the
monomial factor $x^{\delta _{n}}$ - is the image of the Schubert class $%
s_{\lambda }$ under the coproduct of the Verlinde algebra. The analogous
result holds true for osculating walkers using Lem \ref{levelrank}.
\end{proof}

We now invoke the commutation relations of the Yang-Baxter algebra (\ref%
{YBA1}) to provide a closed formula for the above matrix element in terms of
the divided difference operators (\ref{nilHeckerep}).

\begin{proof}[Proof of Corollary \protect\ref{Cor:Demazure}]
Recall Lemma \ref{Lem:ybaDemazure}, then it follows from the definition $%
H(x;q)=A(x)+qD(x)$ that $H(x_{i+1};q)B(x_{i})f=\partial
_{i}B(x_{i+1})H(x_{i};-q)f$ for any $f\in \mathbb{C}[x_{i},x_{i+1},q]\otimes
V^{\otimes N}$ which is symmetric in $x_{i},x_{i+1}$. For our purposes it
suffices to make the stronger assumption that $f$ does not depend on $%
x_{i},x_{i+1}$ and, thus, we simply write $H(x_{i+1};q)B(x_{i})=\partial
_{i}B(x_{i+1})H(x_{i};-q)$ as an operator identity. We prove the following
formula by induction,%
\begin{equation*}
H(x_{n};q)B(x_{n-1})\cdots B(x_{1})=\partial _{n-1}\cdots \partial
_{2}\partial _{1}B(x_{n})\cdots B(x_{2})H(x_{1};(-1)^{n-1}q)\;.
\end{equation*}%
As just discussed Lemma \ref{Lem:ybaDemazure} tells us that the assertion is
true for $n=2$ by setting $i=1$. Assume it holds true for some $n\geq 2$.
Then it follows from Lemma \ref{Lem:ybaDemazure} that%
\begin{multline*}
H(x_{n+1};q)B(x_{n})B(x_{n-1})\cdots B(x_{1})= \\
\partial _{n}B(x_{n+1})H(x_{n};-q)B(x_{n-1})\cdots B(x_{1})= \\
\partial _{n}\partial _{n-1}\cdots \partial _{1}B(x_{n})B(x_{n-1})\cdots
B(x_{2})H(x_{1};(-1)^{n}q)
\end{multline*}%
which completes the proof by induction.

Setting $y_{i}=x_{i+n}$ and $\nabla _{i}=\partial _{n-1+i}\cdots \partial
_{i+1}\partial _{i}$ one now computes 
\begin{multline*}
H(y_{n};q)\cdots H(y_{1};q)B(x_{n})\cdots B(x_{1})= \\
H(x_{2n};q)\cdots H(x_{n+2};q)\nabla_{1}B(x_{n+1})\cdots
B(x_{2})H(x_{1};(-1)^{n}q)= \\
\nabla _{1}H(x_{2n};q)\cdots H(x_{n+2};q)B(x_{n+1})\cdots
B(x_{2})H(x_{1};(-1)^{n}q)= \\
\cdots \\
=\nabla _{1}\cdots \nabla _{n}B(x_{2n})\cdots
B(x_{n+1})H(x_{n};(-1)^{n}q)\cdots H(x_{1};(-1)^{n}q)\;.
\end{multline*}%
Taking the matrix element $\langle \lambda |\ldots |0\rangle $ on both sides
of the last identity and exploiting that $H(x_{i};(-1)^{n}q)|0\rangle
=(1+qx_{i}^{N})|0\rangle $ we find the asserted generating function (\ref%
{toric_generating}) and the identity (\ref{toricDemazure}).
\end{proof}

\begin{example}
Consider $qH^{\ast }(\limfunc{Gr}_{2,4})$. There are six basis elements each
corresponding to a partition in the $2\times 2$ bounding box. Pick any $%
\lambda $ with Young diagram inside the $2\times 2$ bounding box. Invoking (%
\ref{toric_generating}) and (\ref{toricDemazure}) one first computes $%
x^{-\delta _{2}}D_{1}D_{2}\cdots D_{n}y^{\delta _{2}}F_{\lambda }(x;y)$ via (%
\ref{nilHeckerep}) before expanding the result into Schur functions in the $%
x $-variables,%
\begin{equation*}
\begin{tabular}{c||c|c|c|c|c|c|}
$\lambda $ & $\emptyset $ & $(1,0)$ & $(1,1)$ & $(2,0)$ & $(2,1)$ & $(2,2)$
\\ \hline
$s_{\lambda }$ & $1$ & $x_{1}+x_{2}$ & $x_{1}x_{2}$ & $%
x_{1}^{2}+x_{1}x_{2}+x_{2}^{2}$ & $x_{1}^{2}x_{2}+x_{1}x_{2}^{2}$ & $%
x_{1}^{2}x_{2}^{2}$%
\end{tabular}%
\;.
\end{equation*}%
Since the latter form a basis this amounts to solving a linear system of
equations. These steps can be readily implemented on a computer using one's
favourite symbolic computation package such as \emph{Mathematica} or \emph{%
Maple}. The solution yields the toric Schur functions in the $y$-variables
on the left hand side of (\ref{toricDemazure}) which in a similar manner can
be expanded into Schur functions. The table below lists the expansion of the
toric Schur functions $\sum_{d\geq 0}q^{d}s_{\lambda /d/\mu }(y)$ for the
given values of $\lambda $ and $\mu $.%
\begin{equation*}
\begin{tabular}{||c|c|c|c|c|c|}
\hline
$\left. \lambda \right\backslash \mu $ & $(1,0)$ & $(1,1)$ & $(2,0)$ & $%
(2,1) $ & $(2,2)$ \\ \hline
$\emptyset $ & \multicolumn{1}{||c|}{$qs_{2,1}$} & $qs_{2,0}$ & $qs_{1,1}$ & 
$qs_{1,0}$ & $q^{2}s_{2,2}$ \\ \hline
$(1,0)$ & \multicolumn{1}{||c|}{$1+qs_{2,2}$} & $qs_{2,1}$ & $qs_{2,1}$ & $%
qs_{2,0}+qs_{1,1}$ & $qs_{1,0}$ \\ \hline
$(1,1)$ & \multicolumn{1}{||c|}{$s_{1,0}$} & $1$ & $qs_{2,2}$ & $qs_{2,1}$ & 
$qs_{2,0}$ \\ \hline
$(2,0)$ & \multicolumn{1}{||c|}{$s_{1,0}$} & $qs_{2,2}$ & $1$ & $qs_{2,1}$ & 
$qs_{1,1}$ \\ \hline
$(2,1)$ & \multicolumn{1}{||c|}{$s_{2,0}+s_{1,1}$} & $s_{1,0}$ & $s_{1,0}$ & 
$1+qs_{2,2}$ & $qs_{2,1}$ \\ \hline
$(2,2)$ & \multicolumn{1}{||c|}{$s_{2,1}$} & $s_{1,1}$ & $s_{2,0}$ & $%
s_{1,0} $ & $1$ \\ \hline
\end{tabular}%
\end{equation*}%
Since $s_{\lambda /d/\mu }=\sum_{\nu \in (n,k)}C_{\mu \nu }^{\lambda
,d}s_{\nu }$ one can easily read off the Gromov-Witten invariants and the
above table contains the complete multiplication table for $qH^{\ast }(%
\limfunc{Gr}_{2,4})$. For instance, the table entry $\lambda =(1,0)$ and $%
\mu =(2,1)$ reads $qs_{2,0}+qs_{1,1}$ from which we infer $%
C_{(2,1),(2,0)}^{(1,0),1}=1$ and $C_{(2,1),(1,1)}^{(1,0),1}=1$. For more
complicated examples with $n+k>4$ Gromov-Witten invariants greater than one
will occur.
\end{example}

\begin{corollary}
Assume $N=n+k$ with $k\geq n$ and $0<r\leq n$. For any $\lambda ,\mu \in
(n,k)$ we have the following identities for toric Schur functions%
\begin{multline*}
\sum_{\substack{ d+d^{\prime }=r  \\ \mu \in (n,k)}}s_{\lambda /d/\mu
}(x_{1},\ldots ,x_{n})s_{\mu ^{\prime }/d^{\prime }/\nu ^{\prime
}}(-x_{1},\ldots ,-x_{n},x_{n+1},\ldots ,x_{k})= \\
\sum_{d=0}^{r}(-1)^{n(r-d)}e_{r-d}(x_{1}^{N},\ldots ,x_{n}^{N})s_{\lambda
^{\prime }/d^{\prime }/\nu ^{\prime }}(x_{n+1},\ldots ,x_{k},0,\ldots ,0)
\end{multline*}%
which for $n=k$ specialise to the statement in the introduction.
\end{corollary}

\begin{proof}
Consider the matrix element $\langle \lambda |H(x_{1})\cdots
H(x_{n})E(y_{1})\cdots E(y_{k})|\nu \rangle $. Then according to (\ref%
{ZgeneratesGW}) and (\ref{Z'generatesGW}) we have%
\begin{equation*}
\langle \lambda |H(x_{1})\cdots H(x_{n})E(y_{1})\cdots E(y_{k})|\nu \rangle
=\sum_{\substack{ \mu \in (n,k)  \\ d,d^{\prime }\geq 0}}q^{d+d^{\prime
}}s_{\lambda /d/\mu }(x)s_{\mu ^{\prime }/d^{\prime }/\nu ^{\prime }}(y)\;.
\end{equation*}%
On the other hand setting $y_{i}=-x_{i}$ for $i=1,2,\ldots ,n$ we obtain
from the functional relation (\ref{TQ_eqn}) that%
\begin{equation*}
\langle \lambda |H(x_{1})\cdots H(x_{n})E(y_{1})\cdots E(y_{k})|\nu \rangle
=\langle \lambda |E(y_{n+1})\cdots E(y_{k})|\nu \rangle
\prod_{i=1}^{n}(1+(-1)^{n}qx_{i}^{N})\;.
\end{equation*}%
Equating powers of $q$ on both sides we find the claimed identities.
\end{proof}

\end{document}